\author{} 

\title{} 
\date{\today} 

\documentclass[a4paper,11pt]{article}
\usepackage[left=30mm,top=30mm,right=30mm,bottom=30mm]{geometry}
\usepackage{etoolbox} 
\usepackage{booktabs}
\usepackage[table,xcdraw]{xcolor}
\usepackage[usestackEOL]{stackengine}
\usepackage[round]{natbib}
\usepackage[T1]{fontenc}
\usepackage[utf8]{inputenc}
\usepackage{bm}
\usepackage{graphicx}
\usepackage{subcaption}
\usepackage{amsmath}
\usepackage{amsfonts}
\usepackage{mathtools}
\usepackage{xcolor}
\usepackage{float}
\usepackage{hyperref}
\usepackage[capitalise]{cleveref}
\usepackage{enumitem,kantlipsum}
\usepackage{amssymb}
\usepackage{amsbsy}
\usepackage{amsthm}
\usepackage{bbm}
\usepackage{pifont}
 \usepackage{multirow}
 \usepackage{footnote}
 \makesavenoteenv{tabular}
 \usepackage{url}
 \usepackage{eucal}
 \usepackage{wrapfig}
 
\usepackage[ruled,vlined]{algorithm2e}
\usepackage{listings}
\usepackage{dirtytalk}
\usepackage{graphicx}
\usepackage{multirow}

\usepackage{chngcntr}
\usepackage{apptools}
\AtAppendix{\counterwithin{lemma}{section}}
\AtAppendix{\counterwithin{theorem}{section}}
\AtAppendix{\counterwithin{proposition}{section}}
\AtAppendix{\counterwithin{consequence}{section}}

\DeclareMathOperator*{\argmax}{arg\,max}

\newtheorem{proposition}{Proposition}
\newtheorem{definition}{Definition}

\newtheorem{example}{Example}

\newtheorem*{assumption*}{Assumption}

\newtheorem{innercustomthm}{Theorem}
\newenvironment{customthm}[1]
  {\renewcommand\theinnercustomthm{#1}\innercustomthm}
  {\endinnercustomthm}

\newtheorem{innercustomlem}{Lemma}
\newenvironment{customlem}[1]
  {\renewcommand\theinnercustomlem{#1}\innercustomlem}
  {\endinnercustomlem}

\newtheorem{innercustomprop}{Proposition}
\newenvironment{customprop}[1]
  {\renewcommand\theinnercustomprop{#1}\innercustomprop}
  {\endinnercustomprop}

  \newtheorem{theorem}{Theorem}
\newtheorem{lemma}{Lemma}

\newcommand{\indep}{\perp \!\!\! \perp}

\hypersetup{
    colorlinks,
    linkcolor={black},
    citecolor={blue!50!black},
    urlcolor={blue!80!black}
}
\newtheorem{condition}{Condition}
\graphicspath{{figures/}}	

\providecommand{\keywords}[1]
{
  \small	
  \textbf{\textit{Keywords: }} #1
}
\title{Structural restrictions in local causal discovery: identifying direct causes of a target variable} 
\author{Juraj Bodik$^{1 \footnote{ Email of the corresponding author: Juraj.Bodik@unil.ch. \\Published in \textit{Biometrika}, 2025, \href{https://doi.org/10.1093/biomet/asaf042}{10.1093/biomet/asaf042}}}$, Valérie Chavez-Demoulin$^1$}
\date{%
    $^1$ {\small HEC, Université de Lausanne, Switzerland} \\%
    }

\begin{document}

\pagenumbering{gobble}

\maketitle
\begin{abstract}
We consider the problem of learning a set of direct causes of a target variable from an observational joint distribution. Learning directed acyclic graphs (DAGs) that represent the causal structure is a fundamental problem in science. Several results are known when the full DAG is identifiable from the distribution, such as assuming a nonlinear Gaussian data-generating process. Here, we are only interested in identifying the direct causes of one target variable (local causal structure), not the full DAG. This allows us to relax the identifiability assumptions and develop possibly faster and more robust algorithms. In contrast to the Invariance Causal Prediction framework, we only assume that we observe one environment without any interventions. We discuss different assumptions for the data-generating process of the target variable under which the set of direct causes is identifiable from the distribution. While doing so, we put essentially no assumptions on the variables other than the target variable. In addition to the novel identifiability results, we provide two practical algorithms for estimating the direct causes from a finite random sample and demonstrate their effectiveness on several benchmark and real datasets. 
\end{abstract}
  

\pagenumbering{arabic}

\section{Introduction}

Causal reasoning holds great significance in numerous fields, including public policy, decision-making and medicine \citep{Holland, TheBookOfWhy}. Randomized control experiments are widely accepted as the gold standard method for determining causal relationships \citep{Pearl_causal_diagrams_biometrika, Rubin}. However, the feasibility of such experiments is often hindered by their high costs and ethical concerns. In such a case, it is important to estimate causal relations from observational data, which are obtained by observing a system without any interventions \citep{Elements_of_Causal_Inference}. 

A classical approach for causal discovery  is to characterize the Markov equivalence class of structures (MEC, \cite{Meek}); however, the full causal structure is typically unidentifiable. 
A variety of research papers have proposed various methodologies to deal with unidentifiable structures. These methods are either structural-restriction-based, meaning we add some additional assumptions about the functional relations between the variables, such as assuming nonlinear Gaussian data-generation process \citep{hoyer2008, peters2013identifiability,  Chen_biometrika_causal_discovery, Nonlinear_Causal_Discovery_with_Confounders, bodik2023identifiability}; score-based, meaning we pick a causal structure with the best fit on the data according to some score function \citep{Greedy_search, Score-based_causal_learning}; or information-theory-based, using entropy, mutual information and approximations of Kolmogorov complexity \citep{IGCI, Slope, Natasa_Tagasovska}.  However, these methods were designed either for a bivariate case or to infer the entire causal structure of the system. 

In contrast, we focus on inferring only a local structure rather than a global one. This decision is motivated by the expectation that it may lead to simpler, faster, more robust, accurate, and powerful procedures. The assumptions needed for the identifiability are less strict; to be more specific, our methodology relies on so-called local causal sufficiency, an assumption much weaker than the (global) causal sufficiency.

Some methodologies have been proposed to infer a local structure around a target variable. These methods are typically divided into three categories: learning a local skeleton (unoriented graph), learning a minimal Markov blanket (a sufficient set), or learning a set of direct causes of the target variable (goal of this paper). Under causal sufficiency and faithfulness \citep{Pearl_book}, the PC algorithm \citep{PCalgorithm} can identify the MEC and consistently learn the skeleton of the full structure. \cite{Local_causal_discovery_2008}, \cite{Local_Causal_discovery2010},  \cite{Wang2014_local_DAG_learning} and \cite{Ling_local_causal_discovery} discuss modifications of the PC algorithm focusing only on the local structure. 
\cite{Local_Causal_discovery_Gao_2021} suggest a methodology for estimating the minimal Markov blanket based on comparing entropies of the variables.  
\cite{Peters_invariance} introduced an invariance causal prediction method (ICP) for estimating the direct causes of the target variable.
\cite{Local_Causal_discovery_Mona_Azadkia} propose a method to learn the direct causes of the target variable under the assumption that the causes are identifiable (specifically, assuming that the underlying structure is a polytree). In contrast, our work focuses on the identifiability of the direct causes and aims to distinguish between different local causal structures (local MEC) using an structural-restrictions-based approach, where we take the ideas from classical approaches and use them locally.

\subsection{General identifiability of a causal graph}

Consider a DAG (directed acyclic graph) $\mathcal{G}=(V,E)$ with a finite set of vertices (nodes) $V$ and a set of directed edges $E$. We adapt the usual notation of graphical models  (e.g., \citealp{PCalgorithm}); for example, we write $pa_i$, $ch_i$ and $an_i$ for parents, children and ancestors of the node $i$, respectively. Consider a random vector $(X_i)_{i\in V}$ over a probability space $(\Omega, \mathcal{A}, P)$, and we denote $\textbf{X}_S = \{X_s{:}\,\,  s\in S\}$ for $S\subseteq V$. For simplicity, we denote $V = \{0, \dots, p\}$, where $X_0$ (usually denoted by $Y$) is a target variable and $\textbf{X}=(X_1, \dots, X_p)^\top$ are other variables. 

A {structural causal model} (SCM) with a DAG $\mathcal{G}$ over $(X_0,\textbf{X})$ represents a data-generating process where the variables arise from structural equations 
\begin{equation}\label{definition_general_SCM}
    X_i=f_i(\textbf{X}_{pa_i}, \eta_i),\,\,\,\,\,\,\,\,\, f_i\in\mathcal{F}_i,\,\,\,\,\,\,\,i=0,1,\dots, p, 
\end{equation}
where $f_i\in\mathcal{F}_i$ are the assignments (link functions), $\mathcal{F}_i$ are some subsets of measurable functions, and $\eta_i$ are jointly independent random variables. $X_j$ is called a direct cause of $X_i$ if $j\in pa_i$. A set of variables in a SCM is said to be \textit{causally sufficient} if there is no hidden common cause that is causing more than one variable. $\textbf{X}$ is \textit{locally causally sufficient} for $X_0$ if there is no hidden common cause that is causing both $X_0$ and a parent of $X_0$ in $\textbf{X}$. 

We say that $\mathcal{G}$ is \textit{identifiable} under $(\mathcal{F}_0, \dots, \mathcal{F}_p)$ from the joint distribution of $\textbf{X}$ (we also say that the causal model is identifiable) if there is no DAG $\mathcal{G}'\neq \mathcal{G}$ and functions $f_i'\in\mathcal{F}_i, i=0, \dots, p$ generating the same joint distribution. 
We say that the SCM follows an $\mathcal{F}$\textit{-model}, if $\mathcal{F}_i = \mathcal{F}$ for all $i=0, 1, \dots, p$ ; in other words, each structural equation in the SCM satisfies $f_i\in\mathcal{F}, i=0, \dots, p$. 

A large number of $\mathcal{F}$-models were proposed in the literature when the full graph $\mathcal{G}$ is identifiable. \cite{Lingam} show that $\mathcal{G}$ is identifiable under the LiNGaM model (Linear Non-Gaussian additive Models where $\mathcal{F}$ consists of all linear functions and the noise variables are non-Gaussian).
\cite{hoyer2008} and \cite{Peters2014}  developed a framework for additive noise models (ANM) where $\mathcal{F}$ consists of functions additive in the last input, that is,  $X_i = g(\textbf{X}_{pa_i}) +  \eta_i$. 
\cite{Zhang2009} consider the post-nonlinear (PNL) causal model where $\mathcal{F}$ consists of post-additive functions, that is, $
X_i = g_1\big(g_2(\textbf{X}_{pa_i}) +  \eta_i\big)$ with an invertible link function $g_1$.  Note that the former two are special cases of the PNL model. 
 \cite{Khemakhem_autoregressive_flows}, \cite{immer2022identifiability}, and \cite{strobl2022identifying} propose several methods and identifiability results for location-scale models, where  $\mathcal{F}$ consists of location-scale functions, that is, $X_i = g_1(\textbf{X}_{pa_i}) +  g_2(\textbf{X}_{pa_i})\eta_i.$
\cite{bodik2023identifiability} proved the identifiability of $\mathcal{G}$ for a class of conditionally parametric causal models ($CPCM(F)$), where the data-generation process is of the form 
\begin{equation}\label{CPCM_def}
X_i=f_i(\textbf{X}_{pa_i}, \varepsilon_i) = F^{-1}\big(\varepsilon_i; \theta_i(\textbf{X}_{pa_i})\big), \,\,\,\,\,\text{ equivalently } X_i\mid \textbf{X}_{pa_i}\sim F\big(\theta_i(\textbf{X}_{pa_i})\big),
\end{equation}
where $\varepsilon_i\sim U(0,1),$ and $F$ is a known distribution function with parameters $\theta_i\in\mathbb{R}^q$ being functions of the direct causes of $X_i$. Note that if $F$ is Gaussian, then we are in the Gaussian Location-scale Models frame-work, in which case (\ref{CPCM_def}) is equivalent to $X_i = \mu_i(\textbf{X}_{pa_i}) + \sigma_i(\textbf{X}_{pa_i})\eta_i,$ where $\eta_i$ is normally distributed and $\theta_i = (\mu_i, \sigma_i)$.

However, we do not require full identifiability of $\mathcal{G}$ to identify the parents of the target variable. The ICP method \citep{Peters_invariance} does not assume a pre-specified $\mathcal{F}$-model, but rather assumes that the target variable $Y = X_0$ is structurally generated as 
\begin{equation}\label{qwertyuiop}
Y = f_Y(\textbf{X}_{pa_Y}, \eta_Y),\,\,\,\,\,\,\, f_Y \in \mathcal{F}_A, \,\,\,\,\,\,\, \eta_Y \indep \textbf{X}_{pa_Y},
\end{equation}
where $\mathcal{F}_A$ is a space of \textit{additive} functions (the original formulation of ICP \citep{Peters_invariance} was limited to linear functions, but later work extended it to non-linear additive models \citep{Christina}, and more general frameworks have since been proposed \citep{Mooij2020}). The authors additionally assume a multi-environmental setting; that is, we assume that we observe an environmental variable $E$ that is an ancestor of $Y$ but not its parent. This assumption allows for the identification of (a subset of) the parents of $Y$ due to the following invariance property. A set $S \subseteq \{1, \dots, p\}$ is called an $E$-plausible set of causal predictors if $Y \indep E \,| X_{S}$, and 
$$S_E(Y) := \bigcap_{\tilde{S} \subseteq \{1, \dots, p\}:\,\,\tilde{S} \text{ is E-plausible}} \tilde{S}$$ 
is called the $E$-identifiable set of causal predictors (note that we slightly modified the notation from the original paper). It can be shown that always $S_E(Y) \subseteq pa_Y$.  However, $S_E(Y) = pa_Y$ only if $E$ is ``rich'' enough \citep[Section 4]{Peters_invariance}. 

In this work, we do not require the existence of an environmental variable $E$, but only consider one (observational) dataset. Instead, we restrict the functional space of $f_Y$ to achieve the identifiability of the direct causes of $Y$. For example, we demonstrate that in many scenarios, we can identify (a subset of) the parents of $Y$ assuming \textit{only} (\ref{qwertyuiop}). This contributes to innovative theoretical developments that are relevant for causal identifiability.

\subsection{Main idea of our framework}

We assume the data-generation process of $Y$ in the form 
\begin{equation}\label{SCM_for_Y}
Y = f_Y(\textbf{X}_{pa_Y}, \varepsilon_Y), \quad f_Y \in \mathcal{F}, \quad \varepsilon_Y \indep \textbf{X}_{pa_Y}, \quad \varepsilon_Y \sim U(0,1).
\end{equation}
In line with the structural-restrictions-based approach, we assume $f_Y \in \mathcal{F}$, where $\mathcal{F}$ is a subset of all measurable functions (e.g., the class of linear functions). Importantly, we allow for the presence of hidden confounders among the covariates $\textbf{X}$. As opposed to \eqref{definition_general_SCM}, we distinguish between uniformly distributed noise variables (denoted by $\varepsilon_Y$) and arbitrarily distributed ones (denoted by $\eta_Y$). These representations are equivalent under the transformation $q^{-1}(\varepsilon_Y) = \eta_Y$, where $q^{-1}$ denotes the quantile function of $\eta_Y$. Throughout the paper, we adopt the uniform representation without loss of generality. The objective is to estimate the set $pa_Y$ from a random sample of $(Y, \textbf{X})$.

Throughout the paper, we assume the following two assumptions: $f_Y$ is \textit{invertible} and \textit{minimal almost surely}, represented by the notation $f_Y\in\mathcal{I}_m$. Invertibility means that the noise variables can be recovered from the observed variables; that is, a function $f_Y^\leftarrow$ exists such that $\varepsilon_Y = f_Y^\leftarrow(\textbf{X}_{pa_Y},Y).$ Minimality represents the property that there does not exist a lower-dimensional function $g$ and $k\leq m\in\mathbb{N}$ such that $f_Y(x_1, \dots, x_m) = g(x_1, \dots, x_{k-1}, x_{k+1}, \dots, x_m)$ almost surely. The assumption of minimality of a function is closely related to causal minimality. For more details and a rigorous definition of the class of functions $\mathcal{I}_m$, see Appendix \ref{Appendix_A.1.}. Overall, we assume that $f_Y\in\mathcal{F}\subseteq\mathcal{I}_m$. 

Our framework is built on the notion of $\mathcal{F}$-identifiability that we now present.  Recall that, without loss of generality, we assume $\varepsilon_Y\sim U(0,1)$.
\begin{definition}\label{Definition1}
A non-empty set \( S \subseteq \{1, \dots, p\} \) is an \(\mathcal{F}\){-plausible} set of parents of \( Y \) if there exists \( f \in \mathcal{F} \) such that
\[
\varepsilon_S := f^{\leftarrow}(\textbf{X}_S, Y) \quad \text{satisfies} \quad \varepsilon_S \indep \textbf{X}_S \quad \text{and} \quad \varepsilon_S \sim U(0,1).
\]

The set of \(\mathcal{F}\){-identifiable} parents of \( Y \) is defined as
\[
S_\mathcal{F}(Y) := \bigcap_{\substack{S \subseteq \{1, \dots, p\},\, S \neq \emptyset \\ \text{S is } \mathcal{F}\text{-plausible}}} S.
\]
\end{definition}

The constrains on the functional class $\mathcal{F}$ correspond to the data-generation process of $Y$.
If we assume linearity of the covariates, this represents the assumption $f_Y\in\mathcal{F}_L$, where
\begin{equation*}
\mathcal{F}_L = \{f\in\mathcal{I}_m: f(\textbf{x}, \varepsilon) = \beta^T\textbf{x} + q^{-1}(\varepsilon) \text{ for some quantile function }q^{-1} \text{ and } \beta\in\mathbb{R}^{|\textbf{x}|}\}.
\end{equation*}
Note that the restriction $f\in\mathcal{I}_m$ guarantees that the arguments $\beta_i\neq 0$ for all $i$. On the other hand, if we assume that the distribution of $Y\mid \textbf{X}_{pa_Y}$ belongs to a family $F$ (such as Gaussian family), this corresponds to assuming $f_Y\in\mathcal{F}_F$ where 
\begin{equation*}
\mathcal{F}_F := \{f\in\mathcal{I}_m: f(\textbf{x},\varepsilon) =  F^{-1}(\varepsilon;\theta(\textbf{x})) \text{ for some function } \theta\}.
\end{equation*}
We call this restriction conditionally parametric causal model assumption ($CPCM(F)$, see (\ref{CPCM_def})). Table~\ref{tableDefinitions} lists all functional spaces considered in this paper.

\begin{table}[ht]
\centering
\begin{tabular}{|l|}
\hline
\multicolumn{1}{|c|}{Summary of different $\mathcal{F}\subset \mathcal{I}_m$ used in the paper}                                                                                                                    \\ \hline
$\mathcal{F}_L = \{f\in\mathcal{I}_m: f(\textbf{x}, \varepsilon) = \beta^T\textbf{x} + q^{-1}(\varepsilon) \text{ for some quantile function }q^{-1} \text{ and }\beta\in\mathbb{R}^{|\textbf{x}|}\}$                                \\ \hline
$\mathcal{F}_A = \{f\in\mathcal{I}_m: f(\textbf{x}, \varepsilon) = \mu(\textbf{x}) + q^{-1}(\varepsilon) \text{ for some }\mu(\cdot)\text{ and quantile function }q^{-1}     \}$        \\ \hline
$\mathcal{F}_{LS} = \{f\in\mathcal{I}_m: f(\textbf{x}, \varepsilon) = \mu(\textbf{x}) + \sigma(\textbf{x}) q^{-1}(\varepsilon) $, \,\,\,\,\,\ for some function $\mu$,\\
\,\,\,\,\,\,\,\,\,\,\,\,\,\,\,\,\,\,\,\,\,$\text{ positive function }\sigma\,\text{ and a quantile function }q^{-1}\}$\\ \hline
$\mathcal{F}_F := \{f\in\mathcal{I}_m: f(\textbf{x},\varepsilon) =  F^{-1}(\varepsilon;\theta(\textbf{x})) \text{ for some function } \theta:\mathbb{R}^{|\textbf{x}|}\to\mathbb{R}^q\}$         \\ \hline
\end{tabular}
\caption{
The table summarizes different functional spaces \( \mathcal{F} \) used in the paper. 
\( \mathcal{F}_L \), \( \mathcal{F}_A \), \( \mathcal{F}_{LS} \), and \( \mathcal{F}_F \) correspond to the linearity assumption, additivity assumption, location-scale assumption, and CPCM\((F)\) assumption, respectively. 
All classes are subsets of \( \mathcal{I}_m \), meaning their functions are assumed to be invertible with respect to the noise variable.
}
\label{tableDefinitions}
\end{table}

The concept of \(\mathcal{F}\)-identifiability provides theoretical limitations for causal estimates under the assumption \(f_Y \in \mathcal{F}\). The set of \(pa_Y\) can be impossible to estimate (or even ill-defined) even with an infinite number of observations. However, we can consistently estimate the set \(S_{\mathcal{F}}(Y)\) (as discussed in Section~\ref{section_algorithm}). The main focus of this paper is to determine which elements belong to \(S_{\mathcal{F}}(Y)\): When does it hold that \(S_{\mathcal{F}}(Y) = pa_Y\)?
In the following, we present an example to illustrate and clarify the notation and ideas   presented in Section~\ref{Section_3}.

\begin{example}[3 variable case]\label{example_3_variable_case_introduction}
Consider the following structural causal model: $\mathcal{G}$ is in the form $ X_1 \to Y \to X_2$, where $Y$ is generated as $Y = f_Y(X_1, \varepsilon_Y) =  f_0(X_1) + q^{-1}(\varepsilon_Y)$, with $\varepsilon_Y \indep X_1$, for some non-constant function $f_0$ and a quantile function $q^{-1}$. 

Notice that $f_Y \in \mathcal{F}_A$, and $f_Y^{\leftarrow}(X_1, Y) = q(Y-f_0(X_1)) = \varepsilon_Y\indep X_1$. Therefore,  $S = \{1\} = pa_Y$ is $\mathcal{F}_A$-plausible set and we obtain $S_{\mathcal{F}_A}(Y) \subseteq pa_Y$. In Section~\ref{section_algorithm}, we propose an estimator $\hat{S}_{\mathcal{F}_A}(Y)$ that satisfies  $\hat{S}_{\mathcal{F}_A}(Y) \subseteq pa_Y$ with large probability.  

It is important to note that  we do not impose any assumptions on $X_1$ or $X_2$. In Section~\ref{Section_3}, we demonstrate that typically $S_{\mathcal{F}_A}(Y) = pa_Y$ except in some special cases similar to the special cases when ANM is non-identifiable \citep{Zhang2009}. Hence, we can typically identify and consistently estimate the direct causes of $Y$ from a random sample assuming only $f_Y \in \mathcal{F}_A$ and  $\varepsilon_Y \indep \textbf{X}_{pa_Y}$. To the best of our knowledge, there is no similar result in the literature.
\end{example}

Our methodology is quite general, and the scope of potential applications is broad and encompasses a wide range of fields and domains. Assuming additive or heteroskedastic models is a common practice in domains such as gene expressions, economics or biological networks \citep{yuan2006model}.

From a practical point of view, we propose two algorithms for estimating the direct causes of a target variable from a random sample. One provides an estimate of $S_\mathcal{F}(Y)$ with a coverage guarantees; that is, with large probability our estimate is a subset of the parents. Such guarantees are rare and highly desirable in causal discovery. However, the output does not have to contain all direct causes. The second is a score-based algorithm estimating $pa_Y$ based on a goodness-of-fit.

In Section~\ref{Section_3}, we dive deeper into mathematical properties of $S_\mathcal{F}(Y)$, where the aim is to find conditions under which $S_\mathcal{F}(Y) = pa_Y$. In Section~\ref{section_algorithm}, we describe our proposed algorithms for estimating $S_\mathcal{F}(Y)$ and $pa_Y$ from a random sample. Section~\ref{section_simulations} contains a short simulation study followed by an application on a real dataset. The paper has four appendices: Appendix~\ref{Speci_appendix} contains some detailed notions and results omitted from the main text for clarity,  Appendix~\ref{section_appendix_simulations}  contains some details about the simulations and the application, Appendix~\ref{Appendix_Auxiliary} provides some auxiliary results needed for the proofs, and the proofs can be found in Appendix~\hyperref[Section_proofs]{D}.

\section{Properties of \texorpdfstring{$\mathcal{F}$}{}-identifiable parents}
\label{Section_3}
Recall that we assume the data-generation process of $Y$ in the form 
\begin{equation}
Y=f_Y(\textbf{X}_{pa_Y}, \varepsilon_Y),\,\,\,\,\,\,\, f_Y\in\mathcal{F}, \,\,\,\,\,\,\,\varepsilon_Y\indep \textbf{X}_{pa_Y}, \,\,\,\,\,\,\,\varepsilon_Y\sim U(0,1).
\tag{\ref{SCM_for_Y}}
\end{equation}
This assumption implies that $S=pa_Y$ is always  $\mathcal{F}$-plausible set, since $f_Y^{\leftarrow}(\textbf{X}_{pa_Y}, Y)\indep\textbf{X}_{pa_Y} $. Therefore, under \eqref{SCM_for_Y}, it always holds that
\begin{equation}\label{subseteq_parents}
S_\mathcal{F}(Y) \subseteq pa_Y.
\end{equation}
However, the equality $S_\mathcal{F}(Y)= pa_Y$ does not need to hold. Observe that 
$\text{if }\mathcal{F}_1\subseteq \mathcal{F}_2, \text{ then }S_{\mathcal{F}_1}(Y)\supseteq S_{\mathcal{F}_2}(Y).$
This is not surprising, as the more restrictions we put on the data-generation process, the larger the set of identifiable parents. Note that \eqref{SCM_for_Y} inherently assumes local causal sufficiency for $Y$, but full observability of the variables is not required for (\ref{subseteq_parents}) to hold.

The case where \( pa_Y = \emptyset \) needs to be addressed separately since an empty set cannot, by definition, be \(\mathcal{F}\)-plausible. In some untypical situations, such as when the full DAG is non-identifiable, this might even lead to an incorrect conclusion \( S_{\mathcal{F}}(Y) \neq \emptyset = pa_Y \).
This contrasts with the ICP framework \citep{Peters_invariance}, where the case \( pa_Y = \emptyset \) is testable. However, if we use ICP and find \( \widehat{pa}_Y = \emptyset \), we cannot distinguish between two possibilities: 1) \( pa_Y \) is indeed empty, or 2) the environments are not rich enough. Since our framework observes only one environment, we can never reject the second possibility without additional assumptions.

The case \(\mathcal{F} = \mathcal{F}_F\) offers an elegant option for assessing the validity of \( pa_Y = \emptyset \). If $F$ is the marginal distribution of \( Y \) with some (unknown but constant) parameters, we say that \( pa_Y = \emptyset \) is plausible. Although this elegantly corresponds to the \(\mathcal{F}_F\) framework, this does not imply it is a reasonable approach in practice.
Consequently, assuming \( pa_Y \neq \emptyset \) can often be justified either by expert knowledge about the problem or by employing other causal inference methods, such as conducting a do-intervention or orienting certain edges using Meek rules \citep{Meek}.

In this section, we discuss which elements belong to $S_{\mathcal{F}}(Y)$ and we provide various identifiability results under which $S_{\mathcal{F}}(Y) = pa_Y$. We focus on the additive case $\mathcal{F} = \mathcal{F}_A$; however, several counterparts of the shown results for different restrictions such as $\mathcal{F}=\mathcal{F}_{LS}$ and $\mathcal{F}_F$ can be found in Appendix~\ref{Appendix_location_scale_definition}. 
Recall that we interchangeably use the notation $X_0=Y$ for the target variable.

\subsection{Global identifiability \texorpdfstring{$\implies$}{} Local \texorpdfstring{$\mathcal{F}$}{}-identifiability }
\label{section_general_implies_local_identifiability}
In the following, we demonstrate that classical identifiability results from the literature can be used to assess the $\mathcal{F}$-plausibility of a set $S$. Informally, we show that if all variables in the SCM follow an identifiable $\mathcal{F}_A$-model, then any set $S$ containing a child of $Y$ cannot be $\mathcal{F}_A$-plausible. 

To state the result, we use the notion of \textit{restricted additive noise model} (restricted $\mathcal{F}_A$-model), introduced in \citep[Definition 27]{Peters2014}; a submodel of $\mathcal{F}_A$ such that the causal graph is identifiable. It is well known that a bivariate additive noise model is identifiable as long as $(f_j ,P_{X_i})$  does not solve a certain differential equation (leading to non-identifiable cases such as linear Gaussian case \citep{Zhang2009}). A restricted additive noise model consists of such $f\in\mathcal{F}_A$ that does not solve this differential equation for all marginals in the SCM. As an example, one can consider a SCM where $X_j = f_j(\textbf{X}_{pa_j})+\eta_j$, where $\eta_j$ are Gaussian and $f_j$ are non-linear in any component. For more details, see Appendix \ref{Appendix_restricted_additive_noise_model} or   \citep[Section 3.2]{Peters2014}.

For simplicity, we focus on the case when $\textbf{X} = (X_1, \dots, X_p)$ are neighbors (either direct causes or direct effects) of $Y$ in the corresponding SCM. Using the classical conditional independence approach and d-separation \citep{Pearl_causal_diagrams_biometrika}, we can eliminate other variables from being potential parents of $Y$. 

\begin{proposition}
\label{TheoremFidentifiabilityWithChild}
Let $(Y, \textbf{X})$ follow an (identifiable) restricted $\mathcal{F}_A$-model with DAG $\mathcal{G}$, where all $\textbf{X}$ are neighbors of $Y$ in $\mathcal{G}$.  Let $S \subseteq \{1, \dots, p\}$ contain a child of $Y$ in $\mathcal{G}$.  Then, $S$ is not $\mathcal{F}_A$-plausible. 
\end{proposition}

The proof is in \hyperref[proof of TheoremFidentifiabilityWithChild]{Appendix D}. Appendix~\ref{Appendix_pairwise_identifiability} provides an analogous result for general class $\mathcal{F}$.

Following Example~\ref{example_3_variable_case_introduction}, consider $X_2 = f_2(Y) + \eta_2$, where $Y\indep\eta_2$. Combining Proposition~\ref{TheoremFidentifiabilityWithChild} with Theorem~1 in \cite{Zhang2009}, we find that $S=\{2\}$ is not $\mathcal{F}_A$-plausible for a ``typical'' combination of $(f_2, \eta_2)$; for example, if $f_2$ is non-linear and $\eta_2$ has the Gaussian distribution. Therefore, we get $S_{\mathcal{F}_A}(Y)=pa_Y=\{1\}$.

\subsection{Deriving assumptions under which \texorpdfstring{$S_{\mathcal{F}}(Y) = pa_Y$}{Complete set of identifiable parents}}\label{Subsection3.2}

In this section, we explore the $\mathcal{F}$-plausibility of a subset of parents $S\subsetneq pa_Y$. We show that if the function \(f_Y\) is ``sufficiently complicated'', then no subset \(S \subsetneq pa_Y\) is $\mathcal{F}$-plausible.
We focus on the additive case $\mathcal{F} = \mathcal{F}_A$; however, counterparts of the results for $\mathcal{F}_{LS}, \mathcal{F}_{F}$ can be found in Appendix~\ref{Appendix_location_scale_definition}.

\begin{theorem}
\label{Theorem_in_section2}
Let \((Y, \mathbf{X}) \in \mathbb{R} \times \mathbb{R}^p\) be continuous and satisfy \eqref{SCM_for_Y}. Suppose \(\mathcal{F} = \mathcal{F}_A\), and consider a non-empty set \(S \subsetneq pa_Y\). 

\begin{itemize}
 \item (Independent case) Assume that \(\mathbf{X}_{pa_Y}\) has independent components (which can often be achieved through a change of coordinates, for example via independent component analysis \citep{Hyvarinen2001}) and \(f_Y\) is an injective function. Then, the set \(S\) is \(\mathcal{F}_A\)-plausible if and only if \(f_Y\) can be decomposed as follows:  
\begin{equation}\label{decomposition_condition_independent}
        f_Y(\mathbf{x}, e) = h_1(\mathbf{x}_S) + h_2(\mathbf{x}_{pa_Y \setminus S}) + q^{-1}(e), \quad\quad \forall
        \mathbf{x} \in \mathbb{R}^{|pa_Y|}, \, e \in (0, 1),
    \end{equation}
    where \(h_1, h_2\) are measurable functions, \(q^{-1}\) is a quantile function, and \(pa_Y \setminus S = \{i \in pa_Y:i \not\in S\}\).

    \item (General case) The set \(S\) is \(\mathcal{F}_A\)-plausible if and only if the function \(f_Y\) can be expressed as:
    \begin{equation}\label{decomposition_condition}
        f_Y(\mathbf{x}, e) = h_1(\mathbf{x}_S) + h_2(\mathbf{x}) + q^{-1}(e), \quad\quad \forall
        \mathbf{x} \in \mathbb{R}^{|pa_Y|}, \, e \in (0, 1),
    \end{equation}
    for some measurable function \(h_1\), quantile function \(q^{-1}\), and a function \(h_2\in \mathcal{H}_{\mathbf{X}_{pa_Y}}(S)\) where
    \[  \mathcal{H}_{\mathbf{X}_{pa_Y}}(S) := \{f : \mathbb{R}^{|pa_Y|} \to \mathbb{R} \mid f(\mathbf{X}_{pa_Y}) \indep \mathbf{X}_S\}.
    \]

    \item As a consequence, \(S_{\mathcal{F}_A}(Y) = pa_Y\) if and only if:  
    \begin{enumerate}
        \item  \(f_Y\) cannot be expressed in the form of \eqref{decomposition_condition} for any \(S \subsetneq pa_Y\), and  
        \item every set \(S\) that is neither a subset nor a superset of \(pa_Y\) (i.e., \(pa_Y \not\subseteq S \not\subseteq pa_Y\)) is not \(\mathcal{F}_A\)-plausible (e.g., under the assumptions of Proposition~\ref{TheoremFidentifiabilityWithChild}).  
    \end{enumerate}
\end{itemize}
\end{theorem}

\textit{Idea of the proof.}
Full proof is in \hyperref[Proof of Theorem_in_section2]{Appendix D}. Here, we show the main steps of the ''$\implies$'' direction in the case when $p=2$,  $pa_Y = \{1,2\}$ and \( S = \{1\} \). We use a notation  \( Y = f_0(X_1, X_2) + q^{-1}(\varepsilon_Y) \), where \( \varepsilon_Y \indep \textbf{X}_{pa_Y} \). 

\textit{Second bullet-point:} Let \( S = \{1\} \) be an \( \mathcal{F}_A \)-plausible set. That means, there exists \( f \in \mathcal{F}_A \) such that $f^{\leftarrow}(X_1, Y) \indep X_1$. Since \( f \in \mathcal{F}_A \), it has an additive form \( f(x, e) = \mu(x) + \tilde{q}^{-1}(e) \) for \( x \in \mathbb{R}^{|S|} \) and \( e \in (0,1) \), where \( \mu \) is some function and \( \tilde{q}^{-1} \) is a quantile function (strictly increasing due to continuity assumption). Additive functions have an inverse in the form \( f^{\leftarrow}(x, y) = \tilde{q}(y - \mu(x)) \) for \( x \in \mathbb{R}^{|S|} \) and \( y \in \mathbb{R} \) (see the discussion in Appendix~\ref{Appendix_A.1.}). We therefore have:
\begin{equation*}
    \begin{split}
\text{S is } \mathcal{F}_A\text{-plausible}& \iff f^{\leftarrow}(X_1, Y) \indep X_1 \iff Y - \mu(X_1) \indep X_1 \\&
\iff    f_0(X_1, X_2) + q^{-1}(\varepsilon_Y) - \mu(X_1) \indep X_1
 \\&
\iff    f_0(X_1, X_2) - \mu(X_1) \indep X_1.         
    \end{split}
\end{equation*}
Hence,  we directly obtain $f_0(x_1, x_2) - \mu(x_1)\in \mathcal{H}_{\textbf{X}}(S) $, which is the form in \eqref{decomposition_condition}. 

\textit{First bullet-point (case $X_1\indep X_2$):} Fix \( a_0 \in \mathbb{R} \). Since  $f_0(X_1, X_2)  - \mu(X_1) \indep X_1$,  the conditional distribution of \( f_0(X_1, X_2) - \mu(X_1) \mid X_1 = x \) must be the same as \( f_0(X_1, X_2) - \mu(X_1) \mid X_1 = a_0 \) for any \( x \in \mathbb{R} \). Thus, due to the independence, 
\( f_0(x, X_2) - \mu(x) \overset{D}{=} f_0(a_0, X_2) - \mu(a_0).\) By rewriting, we obtain
\begin{equation*}
    f_0(x, X_2) \overset{D}{=} \underbrace{\mu(a_0) - \mu(x)}_{h_1(x)} + \underbrace{f_0(a_0, X_2)}_{h_2(X_2)}.
\end{equation*}
This is almost in the form of \eqref{decomposition_condition_independent}, though the equality holds only in distribution. However, \hyperref[lemma_additivity_equation]{Lemma~C3} extends the result to equality everywhere, provided that an inverse of $f_0$ with respect to $x_2$ exists (which holds under the injectivity assumption). 

\textit{The third bullet-point} is a trivial consequence of Proposition~\ref{TheoremFidentifiabilityWithChild} and the second bullet-point.  
\hfill$\Box$
\newline

Theorem~\ref{Theorem_in_section2} demonstrates that a subset \(S \subsetneq pa_Y\) is an $\mathcal{F}_A$-plausible set if and only if the influence of \(\textbf{X}_{pa_Y}\) on \(Y\) can be decomposed into two independent components, \(h_1\) and \(h_2\). If the function \(f_Y\) is ``sufficiently complicated'', in a sense that this decomposition is not feasible, then no subset \(S \subsetneq pa_Y\) is $\mathcal{F}_A$-plausible. While the function space $\mathcal{H}_{\textbf{X}{pa_Y}}$ is considerably smaller than the space of all possible link functions, its prevalence in real-world scenarios remains unclear. 
\newline

We present additional identifiability results for cases where the support of \( Y \mid \mathbf{X}_S \) is finite. Specifically, for cases when   $\underline{\Psi}(\mathbf{x}_S)\in\mathbb{R}$, where
\[
\underline{\Psi}(\mathbf{x}_S) := \inf \mathrm{supp}(Y \mid \mathbf{X}_S = \mathbf{x}_S) = \inf\{y \in \mathbb{R} : P(Y \leq y \mid \mathbf{X}_S = \mathbf{x}_S) > 0\}, \quad\textbf{x}_S\in \mathrm{supp}(\textbf{X}_S).
\] Then, \( S \) is not \(\mathcal{F}_A\)-plausible if  $Y - \underline{\Psi}(\mathbf{X}_S) \not\indep \mathbf{X}_S.$ A detailed statement with examples is provided in Appendix~\ref{Appendix_support}.
If \( S \subsetneq pa_Y \), this result implies that the only viable candidate for \( h_1 \) in Equation~\eqref{decomposition_condition} is \( h_1 = \underline{\Psi} \) (assuming \( \underline{\Psi} \) is finite). Additionally, if \( S \) includes a child of \( Y \), then \( S \) is typically not \(\mathcal{F}_A\)-plausible (again, assuming \( \underline{\Psi} \) is finite). For further details, refer to Appendix~\ref{Appendix_support}.

\subsection{Issue with linear models }
\label{Section3.1}
The following lemma shows that linear models are not ``sufficiently complicated'' for identifying the parents of $Y$. We use the well-known notion of d-separation, defined in \cite{Pearl_book}.  

\begin{lemma}\label{LemmaAboutUnidentifiabilityFL}
Let $(Y, \textbf{X})\in\mathbb{R}\times \mathbb{R}^p$ follow an $\mathcal{F}_L$-model (linear structural causal model) with DAG $\mathcal{G}_0$ and $pa_Y(\mathcal{G}_0)\neq\emptyset$. Then, $|S_{\mathcal{F}_L}(Y)| \leq 1$. Moreover, if there are any $a,b\in an_Y(\mathcal{G}_0)$ that are d-separated in $\mathcal{G}_0$, then $S_{\mathcal{F}_L}(Y) = \emptyset$. 
\end{lemma}
The proof is in \hyperref[Proof of LemmaAboutUnidentifiabilityFL]{Appendix D}. 
Lemma~\ref{LemmaAboutUnidentifiabilityFL} assumes a causally sufficient model, but this can be relaxed to include hidden variables; see Lemma~\ref{LemmaAboutUnidentifiabilityFL2} in  Appendix~\ref{Appendix_A_Lemma}.

Lemma~\ref{LemmaAboutUnidentifiabilityFL} shows a more general principle that goes beyond the linear models. If we can \textit{marginalize} a causal model to a smaller submodel without breaking $f_Y\in\mathcal{F}$, then only the submodel is relevant for inference about  $S_{\mathcal{F}}(Y)$. 

\begin{lemma}\label{lemma158}
 Let $\mathcal{F}\subseteq\mathcal{I}_m$. Let $(X_0, \textbf{X})\in\mathbb{R}\times \mathbb{R}^p$ follow an $\mathcal{F}$-model with DAG $\mathcal{G}_0$ and $pa_{X_0}(\mathcal{G}_0)\neq \emptyset$. Let  $S\subseteq \{1, \dots, p\}$ be a non-empty set. Let $(X_0, \textbf{X})$ be ``marginalizable'' to $S\cup\{0\}$; that is, $(X_0, \textbf{X}_S)$  can also be written as an $\mathcal{F}$-model with some underlying DAG $\mathcal{G}_S$.  Then $S_{\mathcal{F}}(X_0)\subseteq S$. 
\end{lemma}

\subsection{Hidden confounders}
\label{Section2_hidden_confounders}
We now discuss our framework under the presence of a hidden confounding. Let $\textbf{X} = (\textbf{X}^{obs}, \textbf{X}^{hid})$ denote observed and unobserved covariates, respectively, where $obs\cup hid=\{1, \dots, p\}, obs\cap hid = \emptyset$. 
By definition, the set of $\mathcal{F}$-identifiable parents of $Y$ can be written as
\begin{equation*}
S_\mathcal{F}(Y):= \bigcap_{\substack{S\subseteq obs , S\neq \emptyset\\{ \text{ S is an } \mathcal{F} \text{-plausible}}\\{\text{set of parents of Y}}  }}S.
\end{equation*}
In the presence of a hidden confounder, (\ref{subseteq_parents}) does not need to hold since $\textbf{X}^{hid}$ can create a spurious dependence between \(\varepsilon_Y\) and \(\textbf{X}_{pa_Y \cap obs}\). However, \(S_\mathcal{F}(Y)\) depends on the nature of this dependence, and (\ref{subseteq_parents}) might still hold, as suggested by the following lemma.

\begin{lemma}\label{lemma_hidden_confounder}

 Consider $(Y, \textbf{X})\in\mathbb{R}\times \mathbb{R}^p$  satisfy \eqref{SCM_for_Y} with  $\mathcal{F} = \mathcal{F}_A$. Consider $\emptyset\neq hid\subset pa_Y$. Let $S\subseteq pa_Y \cap obs$ and $\tilde{S}:=(pa_Y\cap obs)\setminus S$ such that
 $(\textbf{X}^{hid}, \textbf{X}_S)\indep \textbf{X}_{\tilde{S}}$ (one can consider that $\textbf{X}^{hid}$ cause $\textbf{X}_S$ and $Y$, but not $\textbf{X}_{\tilde{S}}$). 
 
 If $f_Y$ has a form \begin{equation*}
        f_Y(\textbf{x}, e) = h_1(\textbf{X}^{hid}, \textbf{X}_S) + h_2(\textbf{X}_{\tilde{S}}) + q^{-1}(e), \,\,\,\,\,\textbf{x}\in\mathbb{R}^{|pa_Y|}, e\in(0,1),
     \end{equation*}
for some continuous non-constant real functions  $h_1, h_2$ and a quantile function $q^{-1}$. Then, $S_{\mathcal{F}_A}(Y) \subseteq \tilde{S}\subset pa_Y$.
\end{lemma}

The proof is in \hyperref[Proof of lemma_hidden_confounder]{Appendix D}.  Lemma~\ref{lemma_hidden_confounder} demonstrates that we still obtain $S_{\mathcal{F}_A}(Y) \subseteq pa_Y$ as long as $\textbf{X}^{hid}$ affect $Y$ in a ``sufficiently simple'' way. In this case, ``sufficiently simple'' means that $f_Y$ can be splitted into the part affected by $\textbf{X}^{hid}$ and a part that is not affected by $\textbf{X}^{hid}$. Notice the discrepancy between Theorem~\ref{Theorem_in_section2} and Lemma~\ref{lemma_hidden_confounder}. The function \(f_Y\) should be ``sufficiently complicated'' in order to identify the observed parents but ``sufficiently simple'' to handle hidden confounders. 

\subsection{Non-additive example}\label{subsection3.3}

In many applications, the target variable of interest (e.g., income or salaries) follows a Pareto distribution. In such a case, assuming \( f_Y \in \mathcal{F}_F \) where \( F \) is the Pareto distribution can be a reasonable. This assumption, together with local causal sufficiency, is often enough for identifiability of the direct causes of $Y$. 

\begin{wrapfigure}{r}{2.7cm}
\includegraphics[scale=0.35]{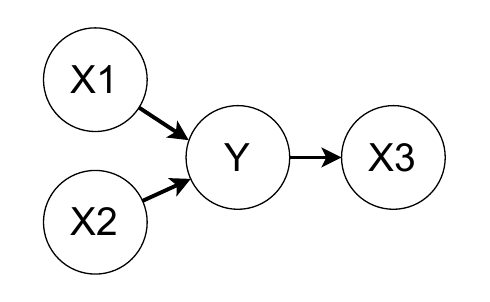}
\label{wfaef}
\end{wrapfigure} 
Consider, for instance, that the (unknown) ground truth is as follows: \( (Y, X_1, X_2, X_3) \) follows an SCM with the DAG on the right. Consider \( Y\mid\textbf{X}_{pa_Y} \) following a Pareto distribution \( F \) with parameter \( \theta(\textbf{X}_{pa_Y}) \). Then \( S_{\mathcal{F}_F}(Y) \subseteq \{1,2\} = pa_Y \).

Additionally, the equality \( S_{\mathcal{F}_F}(Y) = \{1,2\} = pa_Y \) holds under the assumptions of Proposition~\ref{LemmaOParetoinseparabilite} (e.g. if \( \theta(X_1, X_2) = h_1(X_1) + h_2(X_2) \)) and Proposition~\ref{proposition_for_pairwise_F_model_identifiability} (which holds in a  'typical' case, see Consequence~1 in \cite{bodik2023identifiability}).

We have demonstrated that assuming  \eqref{SCM_for_Y} typically allows us to identify the direct causes of $Y$. In the next section, we will show that, in practice, the identifiable direct causes can also be consistently estimated from a random sample.

\newpage

\section{Estimation}
\label{section_algorithm}

\subsection{ISD Algorithm}
We introduce two algorithms for estimating the set of direct causes of a target variable. The first is based on statistical testing of $\mathcal{F}$-plausibility, while the second is a score-based approach. Both methods rely on a random sample of size $n \in \mathbb{N}$ from $(Y, \textbf{X}) \in \mathbb{R} \times \mathbb{R}^p$, where $Y$ denotes the target variable, $\textbf{X}$ the covariates, and $\mathcal{F} \subseteq \mathcal{I}_m$ is the assumed functional class.

For a non-empty set $S\subseteq\{1, \dots, p\}$, define the hypothesis $$H_{0, S} (\mathcal{F}): S \text{   is an }\mathcal{F} \text{-plausible set of parents of Y}. $$ Suppose for the moment that a statistical test for $H_{0, S} (\mathcal{F})$ with size smaller than a significance level $\alpha$ is available. Then, we define $$\hat{S}_{\mathcal{F}}(Y):=\bigcap_{ \emptyset\neq S\subseteq\{1, \dots, p\}:   H_{0, S} (\mathcal{F}) \text{ is not rejected}}S$$ as an intersection of all sets for which  $H_{0, S} (\mathcal{F})$ was not rejected. 

In order to test  $H_{0, S} (\mathcal{F})$,  we propose a procedure called ISD (Independence $+$ Significance $ +$ Distribution). The idea is to decompose  $H_{0, S} (\mathcal{F})$ into three sub-hypothesis. In particular, $H_{0, S} (\mathcal{F})$ is true if and only if there exist a function $\hat{f}_S$ such that $\hat{\varepsilon}_S:=\hat{f}_S^{\leftarrow}(\textbf{X}_S, Y)$  satisfies: 
\begin{enumerate}
    \item   $H_{0, S}^I: \hat{\varepsilon}_S \indep \textbf{X}_S$, \noindent\hfill (\textbf{I}ndependence)
    \item   $H_{0, S}^S$: $\hat{f}_S\in\mathcal{F}$,  \noindent\hfill (\textbf{S}ignificance) \newline (recall that $\hat{f}_S$ must be minimal almost surely or in other words, all inputs are significant)
    \item    $H_{0, S}^D:$ $\hat{\varepsilon}_S \sim U(0,1)$. \noindent\hfill (\textbf{D}istribution)
\end{enumerate}

We reject  $H_{0, S} (\mathcal{F})$ if and only if one of  $H_{0, S}^I, H_{0, S}^S, H_{0, S}^D$  is rejected. 

\begin{theorem}
\label{theorem_consistency_ISD}
Let $(Y, \mathbf{X})$ satisfy \eqref{SCM_for_Y} with $pa_Y \neq \emptyset$.  Assume that the estimator $\hat{S}_{\mathcal{F}}(Y)$ is constructed as described above, using $\hat{f}_{pa_Y} = f_Y$ and valid tests $H_{0, S}^I, H_{0, S}^S, H_{0, S}^D$  for all $S\subseteq \{1, \dots, p\}$ at level $\alpha$ in a sense that for all $S$, $\sup_{P: H_{0, S}^\cdot \text{is true}}P(H_{0, S}^\cdot \text{ is rejected})\leq \alpha$ for all $\cdot \in \{S, I, D\}$. Then
\begin{equation*}
    P(\hat{S}_{\mathcal{F}}(Y) \subseteq pa_Y) \geq 1-3\alpha. 
\end{equation*}

Furthermore, suppose ${S}_{\mathcal{F}}(Y) = pa_Y$, and assume that all tests have non-zero power, i.e., $lim_{n\to\infty}P(H_{0, S}^\cdot \text{ is rejected}\mid  H_{0, S}^\cdot \text{ is false})=1$ for all $\cdot \in \{S, I, D\}$ and all $S\not\supseteq pa_Y$. Then, there exists an integer $n_0$ such that for all $n \geq n_0$, it holds that
\begin{equation*}
    P(\hat{S}_{\mathcal{F}}(Y) = pa_Y) \geq 1-3\alpha. 
\end{equation*}
\end{theorem}
The proof is in \hyperref[Proof of theorem_consistency_ISD]{Appendix D}.  
In order to find a suitable candidate for the function $\hat{f}_S$, we use classical methods from machine learning. 
If $\mathcal{F} = \mathcal{F}_A$ or $\mathcal{F} = \mathcal{F}_{LS}$, we can apply random forest, neural networks, GAM, or other classical methods \citep{GAM,GAMLSS, Paul2014StatisticallyII}. Using one of these methods, we estimate the conditional mean ${\mu}$ (and variance ${\sigma}$ in $\mathcal{F}_{LS}$ case) and output the residuals $\hat{\eta}_S:=Y -\hat{\mu}(\textbf{X}_S)$ (or $ \hat{\eta}_S:=\frac{Y -\hat{\mu}(\textbf{X}_S)}{\hat{\sigma}(\textbf{X}_S)}$ in $\mathcal{F}_{LS}$ case). Possibly, re-scale the residuals $\hat{\varepsilon}_S := \hat{q}(\hat{\eta}_S)$, where $\hat{q}$ is the empirical distribution function of $\hat{\eta}$ (see the discussion about $H_{0, S}^D$ below). If  $\mathcal{F} = \mathcal{F}_F$ for some distribution function $F$, we can use GAMLSS \citep{GAMLSS} or GAM algorithms for estimating $\theta$. Then, we define $\hat{\varepsilon}_S:=F\big(Y, \hat{\theta}(\textbf{X}_S)\big)$.

Notice that if the chosen method is consistent and \eqref{SCM_for_Y} holds, $\hat{f}_{pa_Y}$ converges to $f_{Y}$. Therefore, the choice $\hat{f}_{pa_Y} = f_Y$ in Theorem~\ref{theorem_consistency_ISD} is justified in large sample sizes. The following tests can be used for practical testing of  $H_{0, S}^I, H_{0, S}^S$ and $H_{0, S}^D$:

\begin{enumerate}
    \item $H_{0, S}^I$: We can use a kernel-based HSIC test \citep{Kernel_based_tests} or a copula-based test \citep{copula_based_independence_test}.
    \item $H_{0, S}^S$:  This test ensures that we do not include non-significant (and hence non-causal) covariates into an $\mathcal{F}$-plausible set. In practice, we test the alternative hypothesis \({H}_{0, S}^{S, \text{alt}}: \hat{f}_S \not\in \mathcal{F}\), and we reject \(H_{0, S}^S\) if and only if we do not reject \({H}_{0, S}^{S, \text{alt}}\). The reason is that many methods have been developed for testing \({H}_{0, S}^{S, \text{alt}}\). For example, in the case of linear regression \(Y = \beta \textbf{X}_S + \eta_S\), we test if \(\beta_i \neq 0\) for all \(i \in S\) via classical significance testing. Analogously for GAM or GAMLSS. Alternatively, we can use a permutation test to assess the significance of the covariates \citep{Paul2014StatisticallyII}.
    \item $H_{0, S}^D$: This step is only relevant when a specific noise distribution is assumed. However, the hypothesis \( H_{0, S}^D \) is automatically true in cases such as \(\mathcal{F} = \mathcal{F}_L\), \(\mathcal{F}_A\), or \(\mathcal{F}_{LS}\). In these instances, we omit this test. The reason is that we can use a probability integral transform of the estimated noise to obtain $\hat{\varepsilon}_S\sim U(0,1)$. However in cases such as  $\mathcal{F}=\mathcal{F}_F$, the integral transform breaks the condition $\hat{f}_S\in\mathcal{F}_F$ and testing for $\hat{\varepsilon}_S\sim U(0,1)$ is necessary.  \newline If we opt for testing $H_{0, S}^D:$ $\hat{\varepsilon}_S \sim U(0,1)$, we can use a Kolmogorov-Smirnov or Anderson-Darling test \citep{AD-KStest}.
\end{enumerate}

In our implementation, we opt for HSIC test, GAM estimation, and the Anderson-Darling test. We summarize the algorithm in case of $\mathcal{F} = \mathcal{F}_A$ as follows:

\begin{algorithm}[H]
  \SetAlgoLined
  \KwData{Random sample $(y_1, x_1^1, \dots, x_1^p), \dots, (y_n, x_n^1, \dots, x_n^p)$}
  \KwResult{ REJECT or NOT REJECT  }

1) Estimate \(\hat{f}_S\) in the model \(Y = f_S(\textbf{X}_S) + \eta_S\) (using GAM estimation, for example). Define \(\hat{\eta}_S := Y - \hat{f}_S(\textbf{X}_S)\).

2) Test \(\hat{\eta}_S \indep \textbf{X}_S\) at level \(\alpha\) (using the HSIC test, for example). Set \(H_{0, S}^I = \text{REJECT}\) if this test was rejected, otherwise set \(H_{0, S}^I = \text{NOT REJECT}\).

3) Set \(H_{0, S}^S = \text{NOT REJECT}\) if all covariates are significant at level \(\alpha\) in the model from step 1 (using the permutation test for covariate significance, for example). Otherwise, set \(H_{0, S}^S = \text{REJECT}\).

4) Automatically define \(H_{0, S}^D = \text{NOT REJECT}\) (this step is not relevant in the case \(\mathcal{F} = \mathcal{F}_A\)).

5) Return \(\text{NOT REJECT}\) if all \(H_{0, S}^I\), \(H_{0, S}^S\), and \(H_{0, S}^D\) were not rejected. Otherwise, return \(\text{REJECT}\).

\caption{Testing $H_{0, S} (\mathcal{F})$ in case of $\mathcal{F} = \mathcal{F}_A$ }
  \label{Algorithm}
\end{algorithm}

\subsection{Score-based estimation of \texorpdfstring{$pa_Y$}{}}\label{Section_algorithm2}
We propose a score-based algorithm for estimating the set of direct causes of $Y$. It is a local counterpart of score-based algorithms for estimating the full DAG $\mathcal{G}_0$, following the ideas from \cite{Score-based_causal_learning}, \cite{Peters2014}, and \cite{bodik2023identifiability}. Recall that 
under \eqref{SCM_for_Y}, the set $S=pa_Y$ should satisfy that $\varepsilon_S\indep \textbf{X}_S$, every ${X}_i, i\in S$ is significant and $\varepsilon_S\sim U(0,1)$. 
Therefore, we use the following score function: 
\begin{equation*}
\begin{split}
\widehat{pa}_Y =  \argmax_{\substack{S\subseteq\{1, \dots, p\}\\
                  S\neq\emptyset}}score(S) &= \argmax_{\substack{S\subseteq\{1, \dots, p\}\\
                  S\neq\emptyset}}\lambda_1 (Independence) + \lambda_2(Significance)+ \lambda_3 (Distribution),
\end{split}
\end{equation*}
where $\lambda_1, \lambda_2, \lambda_3\in [0, \infty)$,  ``\textit{Independence}'' is a measure of independence between $ (\hat{\varepsilon}_S, \textbf{X}_S)$, ``\textit{Significance}'' is a measure of significance of covariates $\textbf{X}_S$, and ``\textit{Distribution}'' is a distance between the distribution of $\hat{\varepsilon}_S$ and $U(0,1)$, where $\hat{\varepsilon}_S$ is the noise estimate defined in Section~\ref{section_algorithm}. 

\textit{The measure of independence} can be chosen as the p-value of the independence test (such as the kernel-based HSIC test or the copula-based test). 
\textit{The measure of significance} corresponds to the estimation method analogously to the ISD case. For linear regression (similarly for GAM or GAMLSS), we compute the corresponding p-values for the hypotheses $\beta_i=0, i\in S$.  Then, \textit{Significance} is the minus of the maximum of the corresponding p-values (worst case option). We can also use a permutation test to assess the covariate's significance in terms of the predictability power and choose the largest p-value. 
\textit{The distance between the distribution of} $\hat{\varepsilon}_S$ \textit{and} $U(0,1)$ can be chosen as the p-value of the Anderson-Darling test. 

The choice of $\lambda_1, \lambda_2, \lambda_3$ describes weights we put on each of the three scores: if $\lambda_1>\lambda_2, \lambda_3$, then our estimate will be very sensitive against the violation of the independence $\varepsilon_S\indep \textbf{X}_S$, but not as sensitive against the violation of the other two properties. 

In our implementation, we opt for the following choices. The \textit{Independence} term is the logarithm of the p-value of the Kernel-based HSIC test, and the \textit{Distribution} term is the logarithm of the p-value of the Anderson-Darling test. We use GAM for the estimation of $\hat{f}_S$ and minus the logarithm of the maximum of the corresponding p-values for the \textit{Significance} term. The logarithmic transformation of the three p-values is used to re-scale the values from $[0,1]$ to $(-\infty, 0]$. The practical choice for the weights is   $\lambda_1 = \lambda_2 = \lambda_3 =1 $ (unless $\mathcal{F}=\mathcal{F}_L,\mathcal{F}_A$, or $\mathcal{F}_{LS}$ when we put $\lambda_3 = 0$). 

\subsection{Consistency}

Consistency of the proposed algorithm follows from the results presented in \cite{reviewANMMooij}, who showed consistency of the score-based DAG estimation for additive noise models. In the following, we consider $\mathcal{F} = \mathcal{F}_A$, although it is straightforward to generalize these results for other types of $\mathcal{F}$ (for a discussion about $\mathcal{F} = \mathcal{F}_{LS}$, see \cite{sun2023causeeffect}, and for $\mathcal{F} = \mathcal{F}_F$, see \cite{bodik2023identifiability}). For simplicity, we assume that the measure of independence is the negative value of HSIC test itself (not its p-value as we use in our implementation), and the estimate $\hat{f}_S$ is \textit{suitable} in the sense that noise estimate $\hat{\varepsilon} = \hat{f}_S^\leftarrow(\textbf{X}_S, Y)$ satisfies 
$$
\lim_{n\to\infty}\mathbb{E}_{}\bigg( \frac{1}{n}\sum_{i=1}^n(\varepsilon_i-\hat{\varepsilon}_i)^2   \bigg) = 0,
$$
where the expectation is taken with respect to the distribution of the random sample  \cite[Appendix A.2]{reviewANMMooij}. 

\begin{proposition}\label{Proposition_consistency}
Consider $\mathcal{F} = \mathcal{F}_A$ and 
let $(Y, \textbf{X})\in\mathbb{R}\times \mathbb{R}^p$ follow an SCM with DAG $\mathcal{G}_0$ satisfying \eqref{SCM_for_Y}. Assume that every $S \neq pa_Y\neq \emptyset$ is not $\mathcal{F}$-plausible. 
Then,   

\begin{equation}
  \lim_{n\to\infty} \mathbb{P}(\widehat{pa}_Y  \neq pa_Y) = 0,
\end{equation}   
where $n$ is the size of the random sample and $\widehat{pa}_Y$ is our score-based estimate from Section~\ref{Section_algorithm2} with $\lambda_1, \lambda_2>0, \lambda_3 = 0$, suitable estimation procedure, and HSIC independence measure. 
\end{proposition}

The proof is in \hyperref[Proof of Proposition_consistency]{Appendix D}.  If several sets are $\mathcal{F}$-plausible, the score-based algorithm provides no guarantees that the set $S=pa_Y$ will have the best score among them (as opposed to the ISD algorithm that outputs their intersection). 

\subsection{What is a suitable \texorpdfstring{$\mathcal{F}$}{} in practice?}

The choice of the class $\mathcal{F}$ is a crucial step in the algorithm. The choice of an appropriate model is a common problem in classical statistics; however, it is more subtle in causal discovery. It has been shown \citep{Peters2014} that methods based on restricted structural equation models can outperform traditional methods (these results were shown when estimating the entire DAG). Even if assumptions such as additive noise or Gaussian distribution of the effect given the causes can appear to be strong, such methods have turned out to be rather useful, and small violations of the model still lead to a good estimation procedure. 

The size of  $\mathcal{F}$ is the most important part. If $\mathcal{F}$ contains too many functions, we find that most sets are  $\mathcal{F}$-plausible. On the other hand, overly restrictive  $\mathcal{F}$ can lead to rejecting all sets as potential causes.  If we have knowledge about the data-generation process (such as when $Y$ is a sum of many small events), choosing $\mathcal{F}_F$ for a distribution function $F$ (such as Gaussian) is reasonable. For the choices $\mathcal{F} = \mathcal{F}_A$ or $\mathcal{F}_{LS}$,  there are numerous papers justifying such assumptions in several settings (when estimating the full DAG, \cite{reviewANMMooij}, \cite{Elements_of_Causal_Inference},  \cite{immer2022identifiability}).

\section{Simulation studies and case study}
\label{section_simulations}

We conducted a simulation study to evaluate the performance of the algorithms. These simulations can be found in \hyperref[section_appendix_simulations]{Appendix} \ref{section_appendix_simulations}. The first simulation study (in \hyperref[section_appendix_simulations]{Appendix} \ref{Section_Additive})  demonstrates a congruence between theoretical insights and the outcomes obtained from simulations.  The second simulation study (in \hyperref[section_appendix_simulations]{Appendix} \ref{Section_benchmarks}) involves a comparison between our algorithms and classical approaches using three benchmark datasets.  Table~\ref{Table_results} provides a concise overview of the results.

\begin{table}[h]
\centering
\begin{tabular}{|l|c|c|c|c}
\cline{1-4}
\multirow{2}{*}{}        & \textbf{First}     & \textbf{Second}                         & \textbf{Third}                          & \textbf{Total}  \\
                         & \textbf{Benchmark} & \multicolumn{1}{l|}{\textbf{Benchmark}} & \multicolumn{1}{l|}{\textbf{Benchmark}} & \textbf{(Mean)} \\ \hline
\textbf{IDE algorithm}   & 98\%/ 100\%        & 82\%/ 100\%                             & 72\%/ 100\%                             & 84\%/ 100\%     \\ \cline{1-4}
\textbf{Score algorithm} & 100\%/ 100\%       & 98\%/ 100\%                             & 92\%/ 88\%                              & 93\%/ 96\%      \\ \cline{1-4}
RESIT                    & 52\%/ 18\%         & 36\%/ 100\%                             & 2\%/ 94\%                               & 30\%/ 71\%      \\ \cline{1-4}
CAM-UV                      & 96\%/ 40\%         &  2\%/ 100\%                             & 0\%/ 100\%                              & 32\%/ 80\%      \\ \cline{1-4}
Pairwise bQCD            & 100\%/ 0\%         & 56\%/ 100\%                             & 80\%/ 26\%                              & 78\%/ 42\%      \\ \cline{1-4}
Pairwise IGCI            & 0\%/ 48\%          & 0\%/ 100\%                              & 70\%/ 34\%                              & 17\%/ 70\%      \\ \cline{1-4}
Pairwise Slope           & 100\%/ 2\%         & 100\%/ 100\%                            & 100\%/ 22\%                             & 100\%/ 41\%     \\ \cline{1-4}
\end{tabular}
\caption{Performance of different algorithms on the three benchmark datasets. The first number represents the ``percentage of discovered correct causes'', and the second is the ``percentage of no false positives''. All information can be found in \hyperref[section_appendix_simulations]{Appendix} \ref{section_appendix_simulations}.   }
\label{Table_results}
\end{table}

To illustrate our methodology using a real-world example, we consider data on the fertility rate. This example was also used by \cite{Christina}, who employed the ICP methodology. We show that our results are in line with the findings of  \cite{Christina}, while we drop the assumption of different environments and consider only one environment.

The target variable of interest is $Y = \text{`Fertility rate'},$ measured yearly in more than 200 countries.  Developing countries exhibit a significantly higher fertility rate than Western countries \citep{hirschman1994}. The fertility rate can be predicted by considering covariates such as the `infant mortality rate' or `GDP.' However, if one wants to explore the potential effect of a particular law or a policy change, it becomes necessary to leverage the causal knowledge of the underlying system.

Randomized studies are not possible to design in this context since factors like `infant mortality rate' cannot be isolated for manipulation. Even so, understanding the impact of policies to reduce infant mortality rates within a country remains an important question, even if randomized studies are unfeasible.

Here, we consider covariates $\textbf{X} = (X_1, X_2, X_3, X_4)^\top$, where $X_1=$`GDP (in US dollars)', $X_2 = $`Education expenditure (\% of GDP)', $X_3 = $`Infant mortality rate (infant deaths per 1,000 live births)', $X_4 = $`Continent'. The data are taken from \cite{worldbank_data, unitednations_data_about_fertility}. 

We apply the methodology developed in this paper to estimate the causes of $Y$. Since our variables are continuous and regular, it seems natural and justifiable to use the following choices for $\mathcal{F}$:  $\mathcal{F}_A, \mathcal{F}_{F}$ or $ \mathcal{F}_{LS}$, where $F = Gaussian$. Note that $\mathcal{F}_A, \mathcal{F}_F\subset\mathcal{F}_{LS}$.

For the choice $\mathcal{F} = \mathcal{F}_A$, we observe that all sets $S\subseteq \{1, \dots, 4\}$ are strongly rejected as $\mathcal{F}$-plausible and our estimate is an empty set. Our data show heteroschedasticity and much more complex relations than those that can be described by just one parameter (the mean).  

Applying our methodology with the choices $\mathcal{F}_{LS}$ and $ \mathcal{F}_{F}$, we obtain the results described in Table~\ref{Table_application}. The results suggest that $X_3$ is the identifiable cause of $Y$. This is in line with findings from \cite{Christina} (backed up by research from sociology in \cite{hirschman1994}), who also discovered the variable $X_3$ to be causal. Furthermore, the score-based estimate  indicates that $X_2$ is a member of $\widehat{pa}_Y$ across both selections of $\mathcal{F}$. This suggests that $X_2$ is a cause of $Y$ as well, even though the score-based estimate does not have the same guarantees as the set $\hat{S}_{\mathcal{F}}(Y)$. Note that sets $\{2,3\}, \{1,2,3\}$ are $\mathcal{F}$-plausible for both choices of $\mathcal{F}$. 

Explaining changes in fertility rate is still a topical issue. In our study, we focus on using our developed framework to provide data-driven answers about the potential causes of changes in fertility rates. While models for the fertility rate often have a DAG structure when dynamics are measured, marginalizing to a cross-section may produce relationships which violate the acyclicity constraint \citep{koyama2022how}. This application serves as an illustration of our methodology, while we do not discuss validity of a DAG structure of the variables measured. Moreover the findings rely on the local causal sufficiency of $Y$, an assumption that can surely be questioned. For instance, other variables such as `religious beliefs' or a `political situation' may explain the fertility rate, but are hard to measure.

\begin{table}[h]
\centering
\begin{tabular}{|c|c|c|c|}
\hline
\multirow{2}{*}{$\mathcal{F}$}       & \multirow{2}{*}{$\mathcal{F}$-plausible sets} & ISD estimate of the                                       & Score-based\\
                                     &                                               & $\mathcal{F}$-identifiable set $\hat{S}_{\mathcal{F}}(Y)$ &                                               estimate of $\widehat{pa}_Y$\\ \hline
$\mathcal{F}_{LS}$                   & \{2,3\}, \{3,4\}, \{1,2,3\}, \{1,3,4\}        & \{3\}                                                     & \{2,3\}                                       \\ \hline
$\mathcal{F}_{F}$& \{2,3\}, \{2,3,4\}, \{1,2,3\}, \{1,3,4\}      & \{3\}                                                     & \{1,2,3\}                                     \\ \hline
\end{tabular}
\caption{Estimated causal predictors of fertility rate under different function classes $\mathcal{F}$. Here, $F$ denotes the Gaussian distribution.}
\label{Table_application}
\end{table}

\section{Discussion and future work}

In this work, we studied the problem of estimating the direct causes of a target variable \( Y \). We introduced a general framework that leverages identifiability theory for full causal graphs \( \mathcal{G} \) in a localized setting. This allows for more flexible and scalable applications of causal discovery. 

Several avenues for future work remain. It would be valuable to adapt other causal discovery methods, such as IGCI or those based on Kolmogorov complexity \citep{IGCI, Natasa_Tagasovska}, to the local setting. Similarly, extensions of our framework to time series data \citep{bodik, bodik2024grangercausalityextremes} are a natural next step. Ideas from recent advances in the invariance framework, such as defining sets and simultaneous false discovery bounds \citep{Christina, li2024simultaneousICP}, could also be incorporated into our ISD algorithm to improve its power and computational efficiency. Our local framework is also well-suited for prediction under distribution shift, such as covariate shift or extrapolation \citep{jin2024reweightingpredictiverolecovariate, bodik_extrapolation}, where identifying direct causes enables more robust predictions. Furthermore, leveraging tools such as instrumental variables \citep{Imbens2014} or anchor regression \citep{10.1111/rssb.12398} may offer a principled way to address unobserved confounding and improve robustness in complex settings. 

A central limitation of our approach lies in the choice of the functional class $\mathcal{F}$. As in most causal discovery methods, such as those based on LiNGAM, post-nonlinear models, or location-scale assumptions, the validity and interpretability of the results depend critically on how well the chosen model class matches the underlying data-generating process. While the computational complexity may increase with the dimensionality of $\textbf{X}$, especially when testing many subsets, this also opens opportunities for improvement through scalable algorithms, dimension reduction techniques, or regularization strategies tailored to the local setting.

Overall, the theory developed in this work contributes to a deeper understanding of causal structure and the fundamental limitations of purely data-driven approaches to causal inference. We believe that causal discovery on a local scale provides a promising path toward practical applications, particularly in high-dimensional settings, and this work takes an important step toward making such methods more accessible, interpretable, and robust.

\section*{Acknowledgment}
This study was supported by the Swiss National Science Foundation, grant number 201126.

\section*{Supplementary material}
\label{SM}
The supplementary material contains some detailed definitions and more technical explanations of the theory omitted from the main text for clarity, simulation study, some auxiliary lemmas, and all the technical details and proofs.

The code and data are available in the online repository \url{https://github.com/jurobodik/Structural-restrictions-in-local-causal-discovery.git} or on request from the author.

\bibliography{bibliography}
\newpage
\appendix
\section{Appendix: omitted definitions and technical details} \label{Speci_appendix}
Appendix~\ref{Speci_appendix} contains some detailed definitions and more technical explanations of the theory omitted from the main text. It has 6 parts:

\begin{itemize}
    \item Appendix~\ref{Appendix_A.1.} introduces the concepts of invertibility and minimality, defines a class of invertible causal models, and demonstrates the connection between minimal link functions and causal minimality.
    \item Appendix~\ref{Appendix_A_Lemma} explores unidentifiability in linear structural causal graphs under hidden confounding.
    \item Appendix~\ref{Appendix_restricted_additive_noise_model} provides the definition of a restricted additive noise model from \citep{Peters2014}.
    \item Appendix~\ref{Appendix_pairwise_identifiability} extends Proposition~\ref{TheoremFidentifiabilityWithChild} from Section~\ref{section_general_implies_local_identifiability} to a general functional class $\mathcal{F}$.
    \item Appendix~\ref{Appendix_support} explores $\mathcal{F}$-plausibility when the support of $Y$ is finite. 
    \item Appendix~\ref{Appendix_location_scale_definition} adapts Theorem~\ref{Theorem_in_section2} for non-additive functional classes $\mathcal{F}$.
\end{itemize}

\subsection{Class of invertible and minimal functions, invertible causal model, and causal minimality}
\label{Appendix_A.1.}

First, let us formally introduce the notions of invertibility and minimality of a real function. Next, we define a class of measurable functions $\mathcal{I}_m$ and a subclass of SCM called invertible causal models. We show that minimality of a link function is equivalent with causal minimality of the causal model. 

\begin{definition}[Invertibility]\label{I}
Let $\mathcal{X}_x\subseteq\mathbb{R}^{p},\mathcal{X}_y\subseteq\mathbb{R}, \mathcal{X}_z\subseteq\mathbb{R}$ be measurable sets. A measurable function $f:\mathcal{X}_x\times\mathcal{X}_y\to \mathcal{X}_z$  is called \textbf{invertible for the last element}, notation $f\in \mathcal{I}$, if there exists a function $f^{\leftarrow}:\mathcal{X}_x\times\mathcal{X}_z\to \mathcal{X}_y$ that fulfills the following: $\forall \textbf{x}\in\mathcal{X}_x, \forall y\in\mathcal{X}_y,z\in\mathcal{X}_z$ such that $y=f(\textbf{x},z)$, then $z=f^{\leftarrow}(\textbf{x},y)$. 
\end{definition}
The previous definition indicates that the element $z$ in a relationship $y=f(\textbf{x},z)$ can be uniquely recovered from $(\textbf{x},y)$. To provide an example, for the function $f(x,z) = x+z$, it holds that $f^{\leftarrow}(x, y) = y-x$, since $f^{\leftarrow}(x, f(x,z)) = f(x,z) - x = z$. More generally, for the additive function defined as $f(\textbf{x},z) = g_1(\textbf{x})+ g_2(z)$, where $g_2$ is invertible, it holds that $f^{\leftarrow}(\textbf{x}, y) = g_2^{-1}(y - g_1(\textbf{x}))$, $\textbf{x}\in\mathbb{R}^d, y, z\in \mathbb{R}$. Overall, if $f$ is differentiable and the partial derivative of $f(\textbf{x},z)$ with respect to $z$ is monotonic, then $f\in\mathcal{I}$ (follows from inverse function theorem). 

\begin{definition}[Minimality]
    We say that a function $f:\mathbb{R}^n\to\mathbb{R}$ \textit{is minimal almost surely}, notation $f\in\mathcal{M}$, if there does not exist a function $g:\mathbb{R}^{n-1}\to \mathbb{R}$ and $k\leq n$, such that $f(x_1, \dots, x_n) = g(x_1, \dots, x_{k-1}, x_{k+1}, \dots, x_n)$ for almost all $ \textbf{x}\in\mathbb{R}^n$ in the support of $f$. Recall that the notion 'almost all' represents the fact that the measure of a set $\{\textbf{x}\in\mathbb{R}^n: f(x_1, \dots, x_n) \neq g(x_1, \dots, x_{k-1}, x_{k+1}, \dots, x_n)\}$ has a Lebesgue measure zero. 
\end{definition}

\begin{definition}[Invertible causal model]
 We denote the set of invertible and almost surely minimal functions by
$$
\mathcal{I}_{m} = \{f\in\mathcal{I}\cap \mathcal{M} \}.
$$
We define the \textbf{ICM} (invertible causal model) as a SCM (\ref{definition_general_SCM}) with structural equations satisfying $f_i\in\mathcal{I}_{m}$ for all $i=0, \dots, p$.  
\end{definition}

Note that $f_i\in\mathcal{I}_{m}$ implies causal minimality of the ICM model, as the following lemma suggests. Recall that a distribution $P_\mathbf{X}$ over $\textbf{X}$ satisfies \textbf{causal minimality} with respect to $\mathcal{G}$ if it is Markov with respect to $\mathcal{G}$, but not to any proper subgraph of $\mathcal{G}$. 

\begin{lemma}\label{Lemma_about_ICM_minimality}
Consider a distribution generated by \hyperref[I]{ICM} with graph $\mathcal{G}_0$ (see Definition  \ref{I}). Let all structural equations $f_j\in\mathcal{I}_{m}$, $\forall j=0, \dots, p$.  Then, the distribution is causally minimal with respect to  $\mathcal{G}_0$. Conversely, if $f_j\not \in\mathcal{M}$ for some $j\in\{0, \dots, p\}$, then the causal minimality is violated. 
\end{lemma}
\begin{proof}
The second claim follows directly from Proposition 4 in \cite{Peters2014}. For the first claim, we use a similar approach as in Proposition 17 from \cite{Peters2014}. 

Let $f_j\in\mathcal{I}_{m}$ for all $j=0, \dots, p$ and let the causal minimality be violated, i.e., let $\tilde{\mathcal{G}}$ be a subgraph of $\mathcal{G}_0$ such that the distribution is Markov wrt $\tilde{\mathcal{G}}$. Find $i,j\in\mathcal{G}_0$ such that $i\to j$ in $\mathcal{G}_0$ but $i\not\to j$ in $\tilde{\mathcal{G}}$. 

In graph $\mathcal{G}_0$ we have a structural equation  $X_j = f_j(\textbf{X}_{pa_j(\mathcal{G}_0)}, \varepsilon_j) = f_j(X_i, \textbf{X}_{pa_j(\mathcal{G}_0)\setminus\{i\}}, \varepsilon_j)$  but in $\tilde{\mathcal{G}}$ we have $X_j = \tilde{f}_j(\textbf{X}_{pa_j(\mathcal{G}_0)\setminus\{i\}}, \varepsilon_j)$. Hence, functions $f_j(X_i, \textbf{X}_{pa_j(\mathcal{G}_0)\setminus\{i\}}, \varepsilon_j)$ and $\tilde{f}_j( \textbf{X}_{pa_j(\mathcal{G}_0)\setminus\{i\}}, \varepsilon_j)$ have to be equal almost surely, which contradicts $f_j\in\mathcal{I}_{m}$. 
\end{proof}

\subsection{ Lemma~\ref{LemmaAboutUnidentifiabilityFL} modified for hidden confounders}
\label{Appendix_A_Lemma}

In the following, we use the notion of m-separability, a generalization of d-separability for mixed-type graphs. For details see \cite{Richardson}. Moreover, we say that a node in a graph is a source node, if all edges associated to the node are directed out-going edges (i.e. only $v\to \cdot$ are allowed). 

\begin{lemma}\label{LemmaAboutUnidentifiabilityFL2}
Let $(Y, \textbf{X})\in\mathbb{R}\times \mathbb{R}^p$ follow an $\mathcal{F}_L$-model with DAG ${\mathcal{G}}_0$. Let $\tilde{\textbf{X}}\subseteq \textbf{X}$ be observed variables (and $\textbf{X}\setminus \tilde{\textbf{X}}$ are unobserved hidden confounders). Let $\tilde{\mathcal{G}}_0$ be a projection of ${\mathcal{G}}_0$ on the observed variables. If there exist a source variable $a\in pa_Y(\tilde{\mathcal{G}}_0)$ then $|S_{\mathcal{F}_L}(Y)| \leq 1$. Moreover, if there exist a pair of source variables  $a,b\in an_Y(\tilde{\mathcal{G}}_0)$ that are m-separated in $\tilde{\mathcal{G}}_0$, then $S_{\mathcal{F}_L}(Y) = \emptyset$. 
\end{lemma}

\begin{proof}\label{Proof of LemmaAboutUnidentifiabilityFL 2}
The proof is fully analogous to the proof of Lemma~\ref{Proof of LemmaAboutUnidentifiabilityFL}. Since we added the assumption that $a,b$ are source variables, the fact that some variables are unobserved does not change any step in the proof. 
\end{proof}

\subsection{Definition of restricted additive noise model from \texorpdfstring{\cite{Peters2014}}{}}
\label{Appendix_restricted_additive_noise_model}

We restate the definition of restricted additive noise model presented in Section 3 in \cite{Peters2014}. 

\begin{definition}
An  $\mathcal{F}_A$-model is called a restricted additive noise model if for all \(j \in V\), \(i \in \text{PA}_j\) and all sets \(S \subseteq V\) with \(\text{PA}_j \setminus \{i\} \subseteq S \subseteq \text{ND}_j \setminus \{i, j\}\), there is an \(x_S\) with \(p_S(x_S) > 0\), such that

\[
\left( f_j ( x_{\text{PA}_j \setminus \{i\}}, \underbrace{\cdot}_{X_i} ), P_{(X_i \mid \textbf{X}_S = \textbf{x}_S)}, P_{\eta_j} \right)
\]
satisfies Condition~\ref{Condition19}. Here, the underbrace indicates the input component of \(f_j\) for variable \(X_i\). In particular, we require the noise variables to have non-vanishing densities and the functions \(f_j\) to be continuous and three times continuously differentiable.
\end{definition}

\begin{condition}\label{Condition19}
    The triple \((f_j, P_{X_i}, P_{\eta_j})\) does not solve the following differential equation for all \(x_i, x_j\) with \(\nu''(x_j - f(x_i)) f'(x_i) \neq 0\):

\[
\xi''' = \xi'' \left( -\frac{\nu''' f'}{\nu''} + \frac{f''}{f'} \right) - 2 \nu'' f'' f' + \nu' f''' + \frac{\nu' \nu''' f'' f'}{\nu''} - \nu' \frac{(f'')^2}{f'},
\]
where \(f := f_j\), \(\xi := \log p_{X_i}\), and \(\nu := \log p_{\eta_j}\) are the logarithms of the strictly positive densities. To improve readability, we have skipped the arguments \(x_j - f(x_i)\), \(x_i\), and \(x_i\) for \(\nu\), \(\xi\), and \(f\) and their derivatives, respectively.
\end{condition}

\cite[Theorem 28]{Peters2014} showed that $\mathcal{G}$ is  identifiable from the joint distribution under a causally minimal restricted additive noise model.

\begin{theorem}[Theorem 20 in \cite{Peters2014}]
\label{theorem20}
     Let \(P_{(X_0, X_1)} \) be generated by a bivariate additive noise model with graph \( \mathcal{G}_0 \) satisfying Condition~\ref{Condition19} and assume causal minimality, i.e., a non-constant function \( f_j \). Then, \( \mathcal{G}_0 \) is identifiable from the joint distribution.
\end{theorem}

\begin{theorem}[Theorem 28 in \cite{Peters2014}]
\label{theorem28}
 Let \( P_{(X_1, \ldots, X_p)} \) be generated by a restricted additive noise model with graph \( \mathcal{G}_0 \) and let \( P_{(X_1, \ldots, X_p) }\) satisfy causal minimality with respect to \( \mathcal{G}_0 \) (which holds for example if the functions \( f_j \) are minimal). Then, \( \mathcal{G}_0 \) is identifiable from the joint distribution.

\end{theorem}

\subsection{(Global) identifiability \texorpdfstring{$\implies$}{} (local) \texorpdfstring{$\mathcal{F}$}{}-identifiability for a general non-additive \texorpdfstring{$\mathcal{F}$}{}}
\label{Appendix_pairwise_identifiability}

In the following, we restate the result from Section~\ref{section_general_implies_local_identifiability} for general $\mathcal{F}$. One may expect that $\mathcal{F}$-identifiability of $Y$ follows automatically if we assume an identifiable $\mathcal{F}$-model for all variables in the SCM. This is not the case in general.  
\begin{example}
Consider the following SCM: 
\begin{equation*}
X_1 = \eta_1; \,\,\,\,\,\,\,\,Y = X_1^2 + \eta_Y; \,\,\,\,\,\,\,\,X_2 = \beta_2 Y + \eta_2, 
\end{equation*}
where $\eta_1, \eta_2, \eta_Y$ are independent, $\eta_1\sim U(0,1), \eta_Y, \eta_2\sim N(0, 1)$. This SCM is identifiable (for example, SCM with graph $X_2\to Y\to X_1$ does not allow writing $X_1 = f_1(Y) + \tilde{\eta}_1$ for any $f_1, \tilde{\eta}_1$). However, notice that under conditioning on $X_1 = x\in (0,1)$ we obtain a linear Gaussian case and we can revert the equation between $(Y, X_2)$ and obtain $Y = X_1^2 + \tilde{\beta}_2X_2 + \tilde{\eta}_Y$ for some $\tilde{\beta}_2\in\mathbb{R}, \tilde{\eta}_Y\sim N(0, \sigma^2)$. 
\end{example}

We require a slightly stronger notion of identifiability of $\mathcal{G}$ that we call ``pairwise identifiability''. 

\begin{definition}
Let $(X_0, \textbf{X})\in\mathbb{R}\times \mathbb{R}^p$ follow a SCM (\ref{definition_general_SCM}) with DAG $\mathcal{G}_0$. Let $\mathcal{F}$ be a subset of all measurable functions. We say that the  $\mathcal{F}$-model is \textbf{identifiable}, if there does not exist a graph $\mathcal{G}'\neq \mathcal{G}_0$ and functions $f_i'\in\mathcal{F}, i=0, \dots, p$ generating the same joint distribution. 

We say that the $\mathcal{F}$-model is \textbf{pairwise identifiable}, if for all $i,j\in\mathcal{G}_0, i\in pa_j$ hold the following: $\forall S\subseteq V$ such that  $pa_j\setminus \{i\}\subseteq S \subseteq nd_j\setminus\{i,j\}$ there exist $\textbf{x}_{S}: p_S(\textbf{x}_S)>0$ satisfying that a bivariate model defined as $Z_1=\tilde{\varepsilon}_1, Z_2 = \tilde{f}(Z_1, {\varepsilon}_j)$ is identifiable, where $P_{\tilde{\varepsilon}_1} = P_{X_i\mid \textbf{X}_{S} =\textbf{ x}_S}$, $\tilde{f}(x, \varepsilon) = f(\textbf{x}_{pa_j\setminus\{i\}}, x, \varepsilon)$, $\tilde{\varepsilon}_1\indep {\varepsilon_j}$.  
\end{definition}

In the bivariate case, the notion of identifiability and pairwise identifiability trivially coincides. Note the following observation. 

\begin{lemma}\label{pairwise_implies_global}
Pairwise identifiable $\mathcal{F}$-model is identifiable.  
\end{lemma}
 
The proof can be found in  \hyperref[Proof of pairwise_implies_global]{Appendix D}. The following proposition is a counterpart of Proposition~\ref{TheoremFidentifiabilityWithChild}
 from Section~\ref{section_general_implies_local_identifiability} with general $\mathcal{F}$.

\begin{proposition}
\label{proposition_for_pairwise_F_model_identifiability}
Let $(X_0, \textbf{X})$ follow a pairwise identifiable $\mathcal{F}$-model with DAG $\mathcal{G}$, such that all $\textbf{X}$ are neighbors of $X_0$ in $\mathcal{G}$.   Let $S \subseteq \{1, \dots, p\}$ contain a child of $X_0$ in $\mathcal{G}$.  Then, $S$ is not $\mathcal{F}$-plausible. 
\end{proposition}

\subsection{ \texorpdfstring{$\mathcal{F}$}{}-plausibility under restricted support        }
\label{Appendix_support}

The following proposition discusses a case when $\mathcal{F}$-implausibility results from restricting the support of $Y$ by conditioning on the child of $Y$. This result is specific for an additive and a location-scale space of functions $\mathcal{F}_{A}, \mathcal{F}_{LS}$, but can be easily modified for other types of $\mathcal{F}$.

\begin{proposition}[Assuming bounded support]
\label{Support_proposition}
Let $(Y, \textbf{X})\in\mathbb{R}\times \mathbb{R}^p$ follow an SCM with DAG $\mathcal{G}_0$. Let $S\subseteq\{1, \dots, p\}$ be a non-empty set.

\begin{itemize}
    \item (Additive case) Let  $\underline{\Psi}: \mathbb{R}^{\mid S\mid}\to \mathbb{R}$ be a lower support of $ (Y\mid \textbf{X}_S=\textbf{x}_S)$. Let $\underline{\Psi}$ be finite. 
Then, $S$ is not $\mathcal{F}_{A}$-plausible, if
\begin{equation} \label{eq9987_additive}
Y - \underline{\Psi}(\textbf{X}_S)\not\indep \textbf{X}_S.
\end{equation}

 \item (Location-scale case) Let  $\underline{\Psi},\overline{\Psi}: \mathbb{R}^{\mid S\mid}\to \mathbb{R}$ be real functions such that
\begin{equation*} 
supp(Y\mid \textbf{X}_S=\textbf{x}) = \big(\underline{\Psi} (\textbf{x}),\overline{\Psi}(\textbf{x})\big), \,\,\,\,\,\,\, \forall \textbf{x}\in supp(\textbf{X}_S).
\end{equation*}
Then, $S$ is not $\mathcal{F}_{LS}$-plausible, if
\begin{equation} \label{eq9987}
\frac{Y - \underline{\Psi}(\textbf{X}_S)}{\overline{\Psi}(\textbf{X}_S) - \underline{\Psi}(\textbf{X}_S)}\not\indep \textbf{X}_S.
\end{equation} 
\end{itemize}
\end{proposition}
The proof can be found in \hyperref[Proof of Support_proposition]{Appendix D}. Proposition \ref{Support_proposition} can be expressed as follows. If the support of $Y$ given $\textbf{X}_S = \textbf{x}_S$ is bounded, then $S$ can be  $\mathcal{F}_{LS}$-plausible only in a very specific case when (\ref{eq9987}) does not hold. 
Typically, (\ref{eq9987}) holds if $S$ contains a child of $Y$. 
\begin{example}\label{Example_o_Supporte}
Consider SCM where $Y$ is a parent of $X_1$ and $X_1 = Y + \eta$, where $Y\indep \eta$. Assume that  $Y, \eta$ are non-negative ($supp(Y) = supp(\eta) = (0, \infty)$). 
Then, $\underline{\Psi}(x)=0$ and  $\overline{\Psi}({x})=x$, since the support of $[Y\mid Y+\eta=x]$ is $(0,x)$. Hence, (\ref{eq9987}) reduces to $\frac{Y}{X_1} \not\indep X_1$. If $\frac{Y}{X_1} \not\indep X_1$, then $S=\{1\}$ is not $\mathcal{F}_{LS}$-plausible. 

How strong is the assumption $\frac{Y}{X_1} \not\indep X_1$? We claim that it holds in typical situations. A notable exception when  $\frac{Y}{X_1} \indep X_1$ holds is when $Y, \eta$ have Gamma distributions with equal scales. 
\end{example}

Proposition \ref{Support_proposition} is applicable only when $S$ contains a child of $Y$. If $S\subseteq pa_Y$, then (\ref{eq9987}) typically does not hold, as the following example illustrates.  
\begin{example}
Consider a bivariate SCM with $X_1\to Y$. Let $Y = X_1 + \eta$, where $X_1\indep \eta$. Assume that $supp(X) = supp(\eta) = (0, 1)$. Then, $\underline{\Psi}(x)=x$ and  $\overline{\Psi}({x})=1+x$. Hence, (\ref{eq9987}) reduces to $Y-X_1 \not\indep X_1$, which is not satisfied, so Proposition \ref{Support_proposition} is not applicable. 
\end{example}

\subsection{Theorem~\ref{Theorem_in_section2} restated for location-scale models \texorpdfstring{$\mathcal{F}_{LS}$}{} and CPCM(\texorpdfstring{$F$}{}) models \texorpdfstring{$\mathcal{F}_F$}{}}
\label{Appendix_location_scale_definition}

We restate similar results to Theorem~\ref{Theorem_in_section2} for $\mathcal{F} = \mathcal{F}_{LS}$ and $\mathcal{F}_F$ functional classes. We only focus in the independence case (second bullet-point of Theorem~\ref{Theorem_in_section2}); the general case can be stated analogously. We start with the location-scale result. 

\begin{proposition}[Location-scale]\label{LemmaOLocationScaleinseparabilite}
Let \((Y, \mathbf{X}) \in \mathbb{R} \times \mathbb{R}^p\) be continuous and satisfy \eqref{SCM_for_Y} with \(pa_Y \neq \emptyset\) and $\mathcal{F} = \mathcal{F}_{LS}$. Let $\textbf{X}_{pa_Y}$ have independent components. Then,  $S\subsetneq pa_Y$ is not $\mathcal{F}_{LS}$-plausible if 
 $f_Y$  has the form $$f_Y(\textbf{x}, e)=\mu(\textbf{x}) + \sigma(\textbf{x})e, $$where $\theta(\textbf{x}) = \big(\mu(\textbf{x}), \sigma(\textbf{x})\big)^\top$ is additive in both components, that is,  
$\mu(\textbf{x}) = h_{1, \mu}(x_1)+\dots + h_{k, \mu}(x_k)$ and 
$\sigma(\textbf{x}) =h_{1, \sigma}(x_1)+\dots + h_{k, \sigma}(x_k)$ for some continuous non-constant non-zero functions $h_{i,\cdot}$, where we also assume $h_{i,\sigma}>0$, $i=1, \dots, k$.  
\end{proposition}
Proof can be found in \hyperref[Proof of LemmaOLocationScaleinseparabilite]{Appendix D}. 

Now we focus on the case $\mathcal{F} = \mathcal{F}_{F}$, where a distribution function $F$ has $q\in\mathbb{N}$ parameters $\theta = (\theta_1, \dots, \theta_q)^\top$. We restrict to such $F$ satisfying the following definition.

\begin{definition}\label{DefLS}
Let $F$ be a distribution function with one ($q=1$) parameter $\theta$. We say that the \textbf{parameter acts additively} in $F$, if an invertible real function $f_2$ and a function $f_1\in\mathcal{I}_m$ exist such that for all $\theta_1, \theta_2$ holds \footnote{Notation $F_{\theta_1}\big(F^{-1}_{\theta_2}(z)\big)$ is equivalent to $F(F^{-1}(z, \theta_2), \theta_1)$. We believe that this improves the readability.  } 
\begin{equation}\label{postAdditiveDefinition}
F_{\theta_1}\big(F^{-1}_{\theta_2}(z)\big) = f_1\big(z, f_2(\theta_1) + \theta_2\big), \,\,\,\forall z\in(0,1). 
\end{equation}

We say that the \textbf{parameter acts  multiplicatively} in $F$ if an invertible real function $f_2$ and a function $f_1\in\mathcal{I}_m$ exist such that for all $\theta_1, \theta_2$ holds  \begin{equation}\label{postMultiplDefinition}
F_{\theta_1}\big(F^{-1}_{\theta_2}(z)\big) = f_1\big(z, f_2(\theta_1) \cdot\theta_2\big), \,\,\,\forall z\in(0,1).
\end{equation}

Let $F$ be a distribution function with two ($q=2$) parameters $\theta = (\mu, \sigma)^\top\in\mathbb{R}\times \mathbb{R}_+$. We say that $F$ is a \textbf{Location-Scale} distribution, if for all $\theta$ holds  
\begin{equation*}
F_{\theta}\left( \frac{x-\mu}{\sigma}\right) = F_{\theta_0}(x),\,\,\,\,\forall x\in\mathbb{R},
\end{equation*} 
where $F_{\theta_0}$ is called standard distribution and corresponds to a parameter $\theta_0 = (0,1)^\top$.  
\end{definition}
Examples of $F$ whose parameter acts  additively include a Gaussian distribution with fixed variance or a Logistic distribution/Gumbel distribution with fixed scales.  Note that typically, $f_2(x) = -x$,  since $F_{\theta_1}\big(F^{-1}_{\theta_1}(z)\big)=z$ needs to hold.

Examples of $F$ whose parameter acts  multiplicatively include a Gaussian distribution with the fixed expectation or a Pareto distribution (where $F_{\theta_1}\big(F^{-1}_{\theta_2}(z)\big) = z^{\frac{\theta_1}{\theta_2}}= f_1\big(z, f_2(\theta_1) \cdot\theta_2\big)$ for $f_1(z,x) = z^{-1/x}$ and $f_2(x)=-1/x$). Functions $ f_1, f_2$ are not necessarily uniquely defined. 

Examples of Location-Scale types of distributions include Gaussian distribution, logistic distribution, or Cauchy distribution, among many others.

\begin{proposition}\label{LemmaOParetoinseparabilite}
 Consider continuous $(Y, \textbf{X})\in\mathbb{R}\times \mathbb{R}^p$ satisfying \eqref{SCM_for_Y} with  $pa_Y\neq \emptyset$ and $\mathcal{F} = \mathcal{F}_{F}$. Let $F$ be a distribution function whose parameter acts multiplicatively. 
 Let $\textbf{X}_{pa_Y}$ have independent components. 
\begin{itemize}
\item Consider $f_Y\in\mathcal{F}_F$ in the form $f_Y(\textbf{x}, \varepsilon)=F^{-1}\big(\varepsilon, \theta(\textbf{x})\big)$ with additive function $\theta(x_1, \dots, x_k) = h_1(x_1)+\dots + h_k(x_k)$, where $h_i$ are continuous non-constant real functions. Then, then every $S\subsetneq pa_Y$ is not $\mathcal{F}_{F}$-plausible. 
\item Consider $f_Y\in\mathcal{F}_F$ in the form $f_Y(\textbf{x}, \varepsilon)=F^{-1}\big(\varepsilon, \theta(\textbf{x})\big)$ with multiplicative function   $\theta(x_1, \dots, x_k) = h_1(\textbf{x}_S)\cdot h_2(\textbf{x}_{\{1, \dots, k\}\setminus S})$ for some $S\subsetneq \{1, \dots, k\}$, where $h_1, h_2$ are continuous non-constant non-zero real functions. Then, $S_{\mathcal{F}_F}(Y)=\emptyset$.
\end{itemize}
\end{proposition}
Proof can be found in \hyperref[Proof of LemmaOParetoinseparabilite]{Appendix D}. 
Analogous Proposition~\ref{LemmaOParetoinseparabilite} can be stated for $F$ being additive or location-scale type, where Lemma \ref{CoolLemma} part 3 and 4 would be used instead of part 2.

\setcounter{subsection}{0}
\section{Appendix: Simulations and application}
\label{section_appendix_simulations}

Appendix~\ref{Section_Additive} offers an illustrative  simulations to demonstrate the theoretical findings discussed in Section \ref{Section_3}. Appendix~\ref{Section_benchmarks} pertains to the evaluation of the algorithm's performance on three benchmark datasets. To randomly generate a $d$-dimensional function, we use the concept of the Perlin noise generator \citep{PerlinNoise}. Examples of such generated functions can be found in Appendix~\ref{Appendix_Simulations}. The two algorithms presented in Section~\ref{section_algorithm}, all simulations, and the Perlin noise generator are coded in the programming language \texttt{R} \citep{Rstudio} and can be found in the supplementary package or at \url{https://github.com/jurobodik/Structural-restrictions-in-local-causal-discovery.git}.

\subsection{Highlighting the results from Section \ref{Section_3}}
\label{Section_Additive}

Consider a target variable $Y$ with two parents $X_1, X_2$, where $\textbf{X} = (X_1, X_2)$ is a centered normal random vector with correlation $c\in\mathbb{R}$. The generation process of $Y$ is as follows: \begin{equation}\label{eq576}
    Y = g_1(X_1)+g_2(X_2) + \gamma\cdot g_{1,2}(X_1, X_2)+\eta,\,\,\,\,\,\,\text{ with }\eta\sim N(0,1)\,\text{ and  }\gamma\in\mathbb{R},
\end{equation}
 where $g_1, g_2, g_{1,2}$ are fixed functions generated using the Perlin noise approach.

Theorem~\ref{Theorem_in_section2} suggests that if $c= \gamma= 0 $ then we should find that  $S_{\mathcal{F}}(Y) = \emptyset$. Moreover, if  $c\in\mathbb{R},$ and $\gamma\neq 0 $, then  $S_{\mathcal{F}}(Y) = pa_Y = \{1,2\}$. Moreover, the choice of $c$ can affect the finite sample properties. 

Figure~\ref{Plot_sim_additive} confirms these results. For a range of parameters $c\in[0,0.9], \gamma\in[0,1]$, we generate 50 times such a random dataset of size $n=500$ 
and estimate $S_{\mathcal{F}}(Y)$ using the ISD algorithm. If $\gamma$ is small, we discover direct causes of $Y$ only in a small number of cases. However, the larger the $\gamma$, the larger the number of discovered parents. Figure~\ref{Plot_sim_additive} also suggests that the correlation between the parents can actually be beneficial. The reason is that even if (\ref{eq576}) is additive in each component, the correlation between the components can create a bias in estimating $g_1$ (resp. $g_2$). This results in a dependency between the residuals and $X_1$ (resp. $X_2$) in the model where we regress $Y$ on $X_1$ (or on $X_2$), and we are more likely to reject the plausibility of $S = \{1\}$ (or $S = \{2\}$). 

\begin{figure}[ht]
\centering
\includegraphics[scale=0.5]{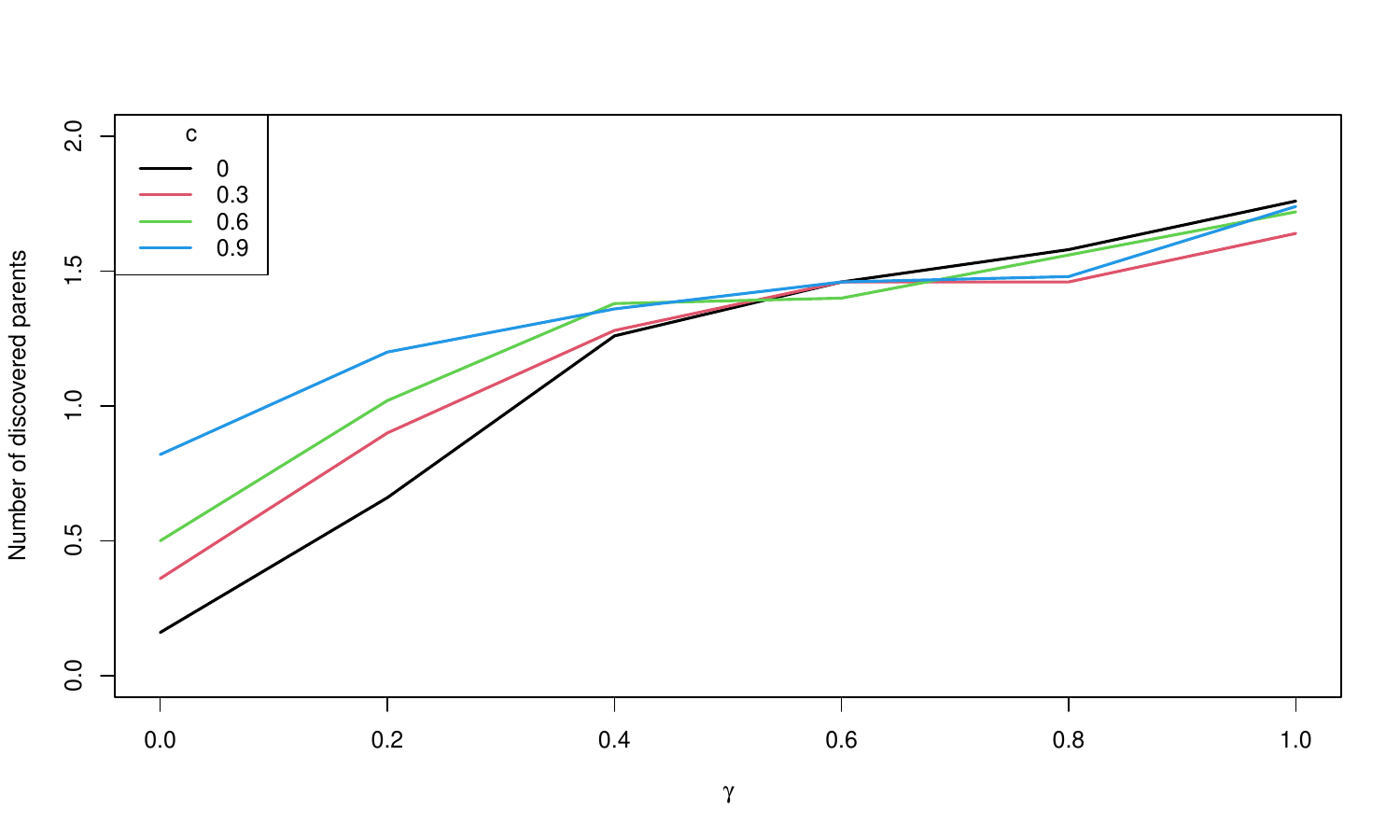}
\caption{Results of the simulation study corresponding to the additive case of Section \ref{Section_Additive}}
\label{Plot_sim_additive}
\end{figure}

\subsection{Benchmarks}\label{Section_benchmarks}
We created three benchmark datasets to assess the performance of our methodology. Two of them correspond to additive noise models ($\mathcal{F}=\mathcal{F}_A$), and the third to $\mathcal{F}=\mathcal{F}_F$ with the Pareto distribution $F$.

The first benchmark dataset consists of $\textbf{X} = (X_1, X_2, X_3, X_4)^\top$ and the response variable $Y$, where $pa_Y = \{1\}$ with the corresponding graph drawn in Figure~\ref{Figure_DAG_for_sim1}A. The data-generation process is as follows: 
\begin{equation*}
 X_1 =\eta_1, \,\,\,Y = g_Y(X_1)+\eta_Y, \,\,\,X_i = g_i(Y,\eta_i),\,\,\, i=2,3,4,     
\end{equation*}
where $g_Y, g_2, g_3, g_4$ are fixed functions generated using the Perlin noise approach, $\eta_1, \eta_2, \eta_3, \eta_4$ are correlated uniformly distributed noise variables, and $\eta_Y\sim N(0,1)$ is independent of $\eta_1, \dots, \eta_4$. The challenge is to find the (one) direct cause among all variables. 
 
The second benchmark dataset consists of $\textbf{X} = (X_1, X_2, X_3)^\top$ and the response variable $Y$, where $pa_Y = \{1,2,3\}$ with the corresponding graph drawn in Figure~\ref{Figure_DAG_for_sim1}B. The data-generation process is as follows:
\begin{align*}
   \begin{pmatrix}
          X_1 \\
            X_2 \\
            X_3
         \end{pmatrix} &\sim N\bigg(\begin{pmatrix}
           0 \\
            0 \\
            0
         \end{pmatrix}, 
         \begin{pmatrix}
          1 ,   c , c \\
            c,   1, c\\
            c,c,1
         \end{pmatrix}\bigg),\,\,\,\, Y = g_Y(X_1, X_2, X_3) + \eta_Y, \,\,\,\text{ where }\eta_Y\overset{}{\sim} N(0,1),
 \end{align*}
for $c=0.5$ and a fixed function $g_Y$ generated using the Perlin noise approach. The challenge is to estimate as many direct causes of $Y$ as possible. 

The third benchmark dataset consists of $\textbf{X} = (X_1, X_2, X_3)^\top$ and the response variable $Y$ corresponding to the DAG C of Figure~\ref{Figure_DAG_for_sim1}. Here, every edge is randomly oriented; either  $\to$ or $\leftarrow$ with probability $\frac{1}{2}$. The source variables (variables without parents) are generated following the standard Gaussian distribution. $Y$ is generated as (\ref{CPCM_def}) with the Pareto distribution $F$ with a fixed function $\theta(\textbf{X}_{pa_Y})$ generated using the Perlin noise approach. Finally, if $X_i$ is the effect of $Y$, it is generated as $X_i = f_i(Y, \eta_i)$, where $\eta_i\sim U(0,1)$,  $ \eta_i\indep Y$ and $f_i$ is a fixed function generated as a combination of functions generated using the Perlin noise approach.

In all datasets, we consider a sample size of $n=500$. 
\begin{figure}[ht]
\centering
\includegraphics[scale=0.55]{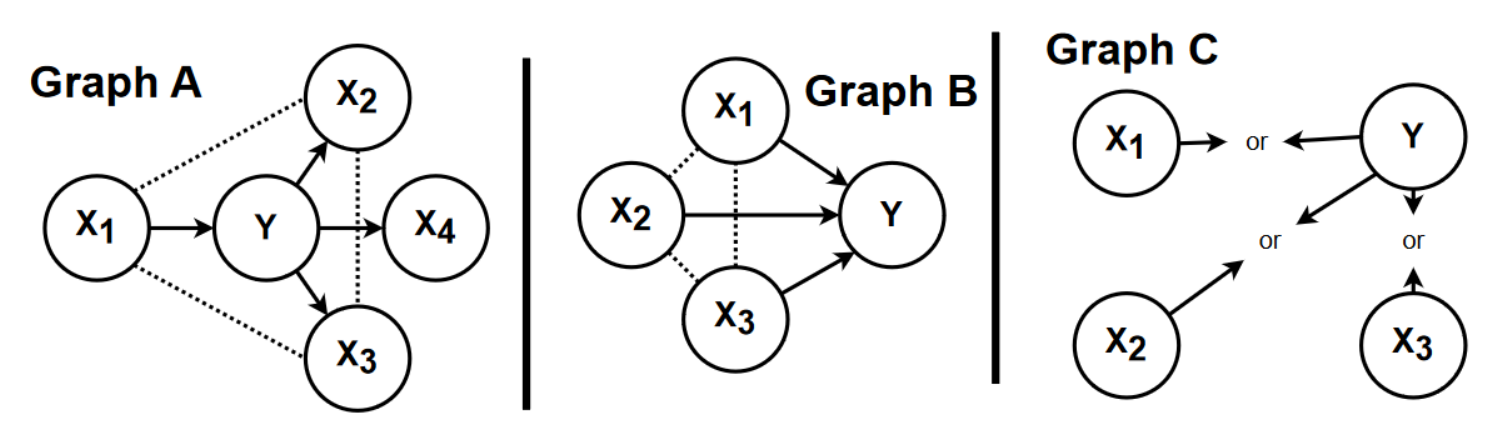}
\caption{Graph A corresponds to the first benchmark dataset, where the noises of $X_1, X_2, X_3$ are correlated (denoted by dashed lines). Graph B corresponds to the second benchmark dataset, where $(X_1, X_2, X_3)$ is generated as a correlated multivariate normally distributed random vector. Graph C corresponds to the third benchmark dataset, where each edge is randomly oriented, either  $\to$ or $\leftarrow$ with probability $\frac{1}{2}$. }
\label{Figure_DAG_for_sim1}
\end{figure}

We compare our proposed algorithms with specific methods for causal discovery, which are: RESIT \citep{Peters2014}, CAM-UV \citep{maeda2021causal}, pairwise bQCD \citep{Natasa_Tagasovska}, pairwise IGCI with the Gaussian reference measure \citep{IGCI}, and pairwise slope \citep{Slope}. When we use the term ``pairwise,'' we are referring to orienting each edge between $(X_i, Y)$ separately, $i=1, \dots, p$. 

For evaluating the performance, we simulate 100 repetitions of each of the three benchmark datasets and use two metrics: ``percentage of discovered correct causes'' and ``percentage of no false positives'' which measures the percentage of cases with no incorrect variable in the set of estimated causes. As an example, consider $pa_Y = \{1,2,3\}$. If we estimate $\widehat{pa}_Y = \{1,2\}$ in $80\%$ of cases and  $\widehat{pa}_Y = \{1,4,5\}$ in $20\%$ of cases, the percentage of discovered correct causes is $\frac{2}{3}\frac{8}{10} + \frac{1}{3}\frac{2}{10} \approx 60\%$ and the percentage of no false positives is $80\%$. 

Table \ref{Table_results} shows the performance of all methodologies. As shown, our two algorithms outperform the other approaches by a significant margin. The IDE algorithm never includes a wrong covariate in the set of causes. On the other hand, although the scoring algorithm demonstrates better overall performance and power, it tends to include non-causal variables more frequently.

\subsection{Visualization of benchmark datasets}
\label{Appendix_Simulations}

In the following, we provide examples of functions generated using the Perlin noise approach. For a one-dimensional case, let $X_1, \eta_Y\overset{iid}{\sim} N(0,1)$ and $Y = g(X_1)+\eta_Y$, where $g$ is generated using the Perlin noise approach. Such (typical) datasets are plotted in Figure~\ref{Figure_perlin_1}. 

For a two-dimensional case, let $X_1,X_2, \eta_Y\overset{iid}{\sim} N(0,1)$ and $Y = g(X_1, X_2)+\eta_Y$, where $g$ is generated using the Perlin noise approach. Such (typical) datasets are plotted in Figure~\ref{Figure_perlin_2}. 

For a three-dimensional case, let $X_1,X_2,X_3, \eta_Y\overset{iid}{\sim} N(0,1)$ and $Y = g(X_1, X_2,X_3)+\eta_Y$, where $g$ is generated using the Perlin noise approach.
\begin{figure}[ht]
\centering
\includegraphics[scale=0.6]{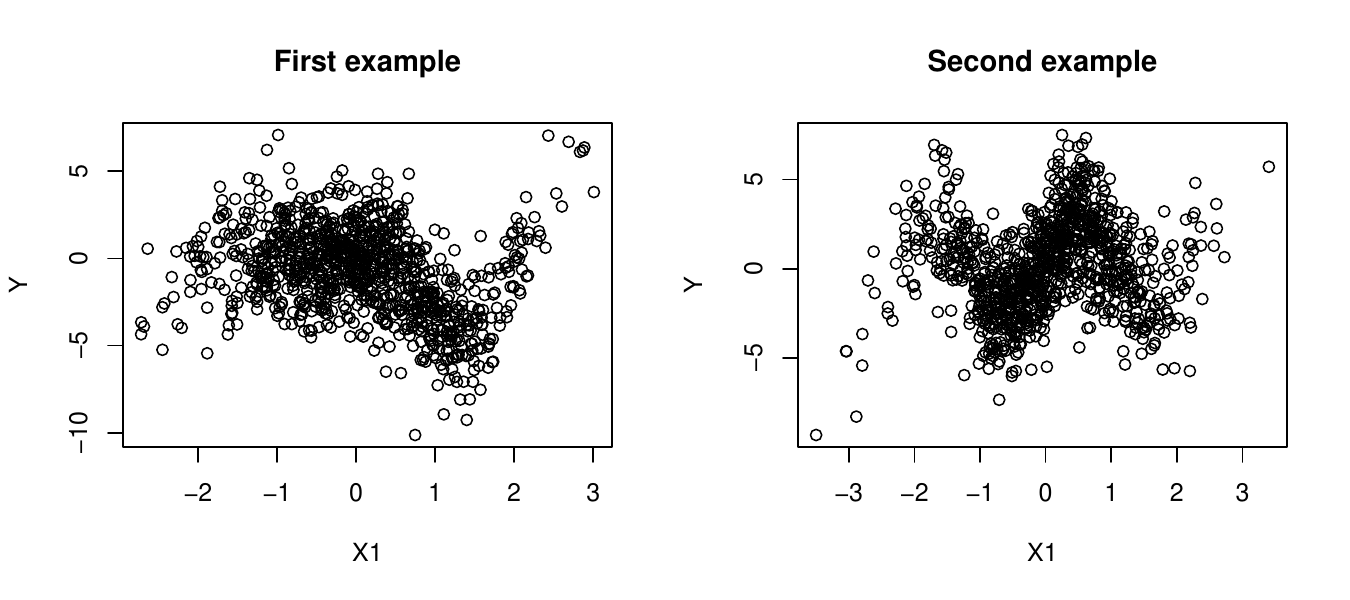}
\caption{Two examples of one dimensional functions generated using the Perlin noise approach}
\label{Figure_perlin_1}
\end{figure}

\begin{figure}[ht]
\centering
\includegraphics[scale=0.3]{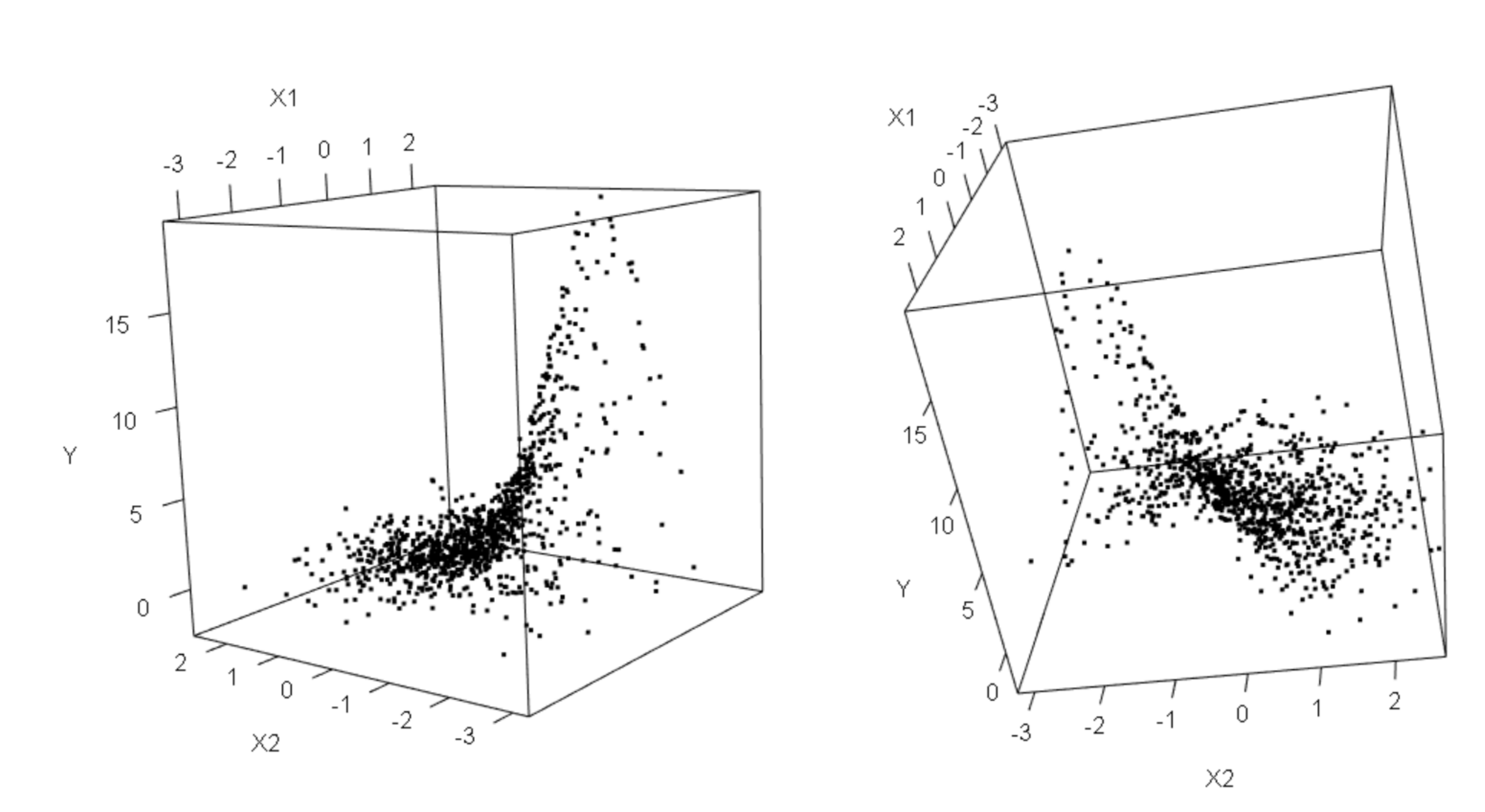}
\caption{A typical two-dimensional function generated using the Perlin noise approach, shown from two different angles }
\label{Figure_perlin_2}
\end{figure}

\section{Appendix:  Auxiliary results}
\label{Appendix_Auxiliary}

\begin{lemma}\label{distributionalequalitylemma}
Let $X$ be a non-degenerate continuous real random variable. Let $a,b\in\mathbb{R}$ such that 
\begin{equation}\label{qwerty}
a+bX\overset{D}{=}X.
\end{equation}
Then, either $(a,b) = (0,1)$ or $(a,b) = (2med(X), -1)$. Here, $med(X)$ is the median of $X$.
\end{lemma}
\begin{proof}
\textit{Idea of the proof assuming a finite variance of $X$: } If $X$ has finite variance, then (\ref{qwerty}) implies $var(a+bX) = var(X)$, rewriting gives us  $b^2 var(X)=var(X)$, and hence, $b=\pm 1$. Now, (\ref{qwerty}) also implies $\mathbb{E}(a+bX) = \mathbb{E}(X)$, hence $a=(1-b)\mathbb{E}(X)$. Therefore, if $b=1$, then $a= 0$, and if $b=-1$, then $a=2\mathbb{E}(X)$. 

\textit{Proof without the moment assumption:}  (\ref{qwerty}) implies that for any $q\in (0.5,1)$, the difference between the $q$ quantile and $(1-q)$ quantile should be the same on both sides of (\ref{qwerty}). Denote $F^{-1}_X(q)$ a $q-$quantile of $X$ and assume that $F^{-1}_X(q)\neq F^{-1}_X(1-q)$ (since $X$ is non-degenerate, such $q$ exist). We get 
$$
F^{-1}_{a+bX}(q) - F^{-1}_{a+bX}(1-q) = F^{-1}_X(q)- F^{-1}_X(1-q)=:D.
$$
Consider $b\geq 0$. Using linearity of the quantile function, we obtain $a+bF^{-1}_X(q) - \big(a+bF^{-1}_X(1-q)\big) = D$ and hence, $bD=D$, which gives us $b=1$. If $b<0$, then an identity $F^{-1}_{a+bX}(q) = a+\big(1-F^{-1}_{-bX}(1-q)\big) = a+\big(1+bF^{-1}_{X}(1-q)\big)$ hold. Hence, we get  $a+[1+bF^{-1}_X(1-q)] - [a+\big(1+bF^{-1}_X(1-q)\big)] = D$. Rewriting the left side, we get $-bD=D$,  which gives us $b=-1$. 

In the case when $b=1$, trivially $a=0$, since otherwise, $med(a+X)\neq med(X)$. If $b=-1$, then applying median on both sides of (\ref{qwerty}) gives us $med(a-X)= med(X)$ and hence, $a=2med(X)$, as we wanted to show. 
\end{proof}

\begin{lemma}\label{CoolLemma}
Let $\textbf{X}=(X_1, \dots, X_k)$ be a continuous random vector with independent components and $s<k$. 
\begin{enumerate}
\item Let $f:\mathbb{R}^{k}\to\mathbb{R}$ be an injective  function such that there does not exist a decomposition  $ f(\textbf{x}) = f_1(\textbf{x}_S) + f_2(\textbf{x}_{\{1, \dots, k\}\setminus S}), \textbf{x}\in\mathbb{R}^{k}$ for any non-empty $S\subset \{1, \dots, k\}$, where $f_1, f_2$ are some measurable functions. 

 Then, a measurable function $h$ does not exist such that for $s<k$ holds
\begin{equation}\label{yuiMAIN}
f(X_1, \dots, X_k)+  h(X_1, \dots, X_s) \indep (X_1, \dots, X_s).
\end{equation}

\item  Let $f_1, \dots, f_k$ be continuous non-constant real functions. A non-zero function $h$ does not exist such that 
\begin{equation}\label{yui}
  h(X_1, \dots, X_s)\big(f_1(X_1)+\dots+f_k(X_k)\big)\indep (X_1, \dots, X_s).
\end{equation}
\item  Let $f_1, \dots, f_k$ be continuous non-constant non-zero real functions. Then, a non-zero function $h$ does not exist such that 
\begin{equation}\label{yuiDVA}
 h(X_1, \dots, X_s) + f_1(X_1)f_2(X_2)\dots f_k(X_k)\indep (X_1, \dots, X_s).
\end{equation}
\item Let $f:\mathbb{R}^{k-s}\to\mathbb{R}$ be measurable function such that $f(X_{s+1}, \dots, X_k)$ is non-degenerate continuous random variable. Functions $h_1, h_2$ does not exist, such that $h_2$ is positive non-constant and
\begin{equation}\label{yuiTRI}
h_1(X_1, \dots, X_s) + h_2(X_1, \dots, X_s)f(X_{s+1}, \dots, X_k)\indep (X_1, \dots, X_s).
\end{equation}

\end{enumerate}
\end{lemma}
\begin{proof}
We use notation $\textbf{X}_S=(X_1, \dots, X_s)^\top$, $\textbf{X}_{\setminus S}=(X_{s+1},\dots X_k)^\top$. 
Let us introduce functionals (not norms, we only use them to simplify notation) $||\cdot||_{plus}$ and  $||\cdot||_{times}$, defined by $|| \textbf{a} ||_{plus} = a_1 +\dots + a_d$, $|| \textbf{a} ||_{times} = a_1 a_2\dots a_d$, for $\textbf{a}=(a_1, \dots, a_d)^\top\in\mathbb{R}^d$. 

\textbf{Part 1:} For a contradiction, let such $h$ exist. Define $\xi:=h(\textbf{X}_S) + f(\textbf{X}_S, \textbf{X}_{\setminus S}) $, which is the left hand side of (\ref{yuiMAIN}). Fix $\textbf{a}_0\in\mathbb{R}^s$ in the support of $\textbf{X}_S$ and define
$$f_1(\textbf{x}):=h(\textbf{a}_0) - h(\textbf{x}), \,\,\text{for}\,\,\textbf{x}\in\mathbb{R}^{s}, \,\,\,\,\,and\,\,\,\,\,\, f_2(\textbf{x}) := f(\textbf{a}_0, \textbf{x})\,\,\text{for}\,\,  \textbf{x}\in\mathbb{R}^{k-s}.   $$

Since $\xi\indep \textbf{X}_S$, for all $\textbf{x}\in\mathbb{R}^{s}$ holds   $\xi\mid [\textbf{X}_S=\textbf{a}_0] \overset{D}{=}\xi\mid [\textbf{X}_S=\textbf{x}]$. Hence, 
\begin{align*}
h(\textbf{x}) + f(\textbf{x}, \textbf{X}_{\setminus S})\overset{D}{=}h(\textbf{a}_0) + f(\textbf{a}_0, \textbf{X}_{\setminus S})\end{align*}
\begin{equation}
    \label{awrefarefs}
    f(\textbf{x}, \textbf{X}_{\setminus S})\overset{D}{=} f_1(\textbf{x}) +f_2(\textbf{X}_{\setminus S}). 
\end{equation}
To extend the equality from equality in distribution to equality everywhere, we use Lemma~\ref{lemma_additivity}. We found an additive decomposition of $f$, which is the desired contradiction. 

\textbf{Part 2: }For a contradiction, let such $h$ exist. First, some notation: Let $Y=f_{s+1}(X_{s+1})+\dots+f_{k}(X_k)$ and define $\xi:=h(\textbf{X}_S)(||f_S(\textbf{X}_S)||_{plus}+Y)$, where  $f_S: \mathbb{R}^s\to\mathbb{R}^s:f_S(\textbf{x})= (f_1(x_1), \dots, f_s(x_s))^\top$, which is the left hand side of (\ref{yui}). 

Choose $\textbf{a}, \textbf{b}, \textbf{c}\in\mathbb{R}^s$ in the support of $\textbf{X}_S$ such that $||f_S(\textbf{a})||_{plus}, ||f_S(\textbf{b})||_{plus}, ||f_S(\textbf{c})||_{plus}$ are distinct and $h(\textbf{b})\neq 0$ (it is possible since $h_i$ are non-constant). 

Since $\xi\indep \textbf{X}_S$, then  $\xi\mid [\textbf{X}_S=\textbf{a}] \overset{D}{=}\xi\mid [\textbf{X}_S=\textbf{b}] \overset{D}{=}\xi\mid [\textbf{X}_S=\textbf{c}]$. Hence, 
\begin{equation}\label{asdfgh}
h(\textbf{a})(||f_S(\textbf{a})||_{plus}+Y)\overset{D}{=}h(\textbf{b})(||f_S(\textbf{b})||_{plus}+Y)\overset{D}{=}h(\textbf{c})(||f_S(\textbf{c})||_{plus}+Y).
\end{equation}
By dividing by a non-zero constant $h(\textbf{b})$ and subtracting a constant $||f_S(\textbf{b})||_{plus}$, we get
$$
\frac{h(\textbf{a})}{h(\textbf{b})}||f_S(\textbf{a})||_{plus}-||f_S(\textbf{b})||_{plus}+\frac{h(\textbf{a})}{h(\textbf{b})}Y\overset{D}{=}Y\overset{D}{=}\frac{h(\textbf{c})}{h(\textbf{b})}||f_S(\textbf{c})||_{plus}-||f_S(\textbf{b})||_{plus}+\frac{h(\textbf{c})}{h(\textbf{b})}Y.
$$
Now we use Lemma \ref{distributionalequalitylemma}. It gives us that $\frac{f(\textbf{a})}{f(\textbf{b})}=\pm 1$ and also  $\frac{f(\textbf{c})}{f(\textbf{b})}=\pm 1$. Therefore, at least two values of $f(\textbf{a}), f(\textbf{b}), f(\textbf{c})$ must be equal (and neither of them are zero). WLOG $f(\textbf{a})= f(\textbf{c})$. Plugging this into equation (\ref{asdfgh}), we get $||h_S(\textbf{a})||_{plus}=||h_S(\textbf{c})||_{plus}$, which is a contradiction since we chose them to be distinct. 

\textbf{Part 3: } We proceed in a similar way to the previous part. For a contradiction, let such $h$ exist. First, some notation: let $Y=f_{s+1}(X_{s+1})\dots f_{k}(X_k)$ and define $\xi:=h(\textbf{X}_S) + (||f_S(\textbf{X}_S)||_{times} \cdot Y)$, where  $f_S: \mathbb{R}^s\to\mathbb{R}^s: f_S(\textbf{x})= (f_1(x_1), \dots, f_s(x_s))^\top$, which is the left hand side of (\ref{yuiDVA}). 

Choose $\textbf{a}, \textbf{b}, \textbf{c}\in\mathbb{R}^s$ in the support of $\textbf{X}_S$ such that $||f_S(\textbf{a})||_{times}, ||f_S(\textbf{b})||_{times}, ||f_S(\textbf{c})||_{times}$ are distinct and $||f_S(\textbf{b})||_{times}\neq 0$. 

Since $\xi\indep \textbf{X}_S$, then  $\xi\mid [\textbf{X}_S=\textbf{a}] \overset{D}{=}\xi\mid [\textbf{X}_S=\textbf{b}] \overset{D}{=}\xi\mid [\textbf{X}_S=\textbf{c}]$. Hence,

\begin{equation}\label{asdfghDVA}
h(\textbf{a}) + ||f_S(\textbf{a})||_{times}\cdot Y\overset{D}{=}h(\textbf{b})+||f_S(\textbf{b})||_{times}\cdot Y\overset{D}{=}h(\textbf{c}) + ||f_S(\textbf{c})||_{times}\cdot Y.
\end{equation}
By dividing by a non-zero constant $||f_S(\textbf{b})||_{times}$ and subtracting constant $h(\textbf{b})$, we get
$$
h(\textbf{a}) - h(\textbf{b}) + \frac{||f_S(\textbf{a})||_{times}}{||f_S(\textbf{b})||_{times}}Y \overset{D}{=}Y\overset{D}{=} h(\textbf{c}) - h(\textbf{b}) + \frac{||f_S(\textbf{c})||_{times}}{||f_S(\textbf{b})||_{times}}Y
$$

Now we use lemma \ref{distributionalequalitylemma}. It gives us that $\frac{||f_S(\textbf{a})||_{times}}{||f_S(\textbf{b})||_{times}}=\pm 1$ and also  $\frac{||f_S(\textbf{c})||_{times}}{||f_S(\textbf{b})||_{times}}=\pm 1$. Therefore, at least two values of $||f_S(\textbf{a})||_{times}, ||f_S(\textbf{b})||_{times}, ||f_S(\textbf{c})||_{times}$ must be equal, which is a contradiction since we chose them to be  distinct. 

\textbf{Part 4: } For a contradiction, let $h_1, h_2$ exist. Denote $Y = f(X_{s+1}, \dots, X_k)$. Choose $\textbf{a}, \textbf{b}\in\mathbb{R}^s$ in the support of $\textbf{X}_S$ such that $h_2(\textbf{a})\neq h_2(\textbf{b})\neq 0$. From (\ref{yuiTRI}), we get $h_1(\textbf{a}) + h_2(\textbf{a})Y \overset{D}{=}h_1(\textbf{b}) + h_2(\textbf{b})Y$. By rewriting, we get $\frac{h_1(\textbf{a})-h_1(\textbf{b})}{h_2(\textbf{b})} + \frac{h_2(\textbf{a})}{h_2(\textbf{b})}Y \overset{D}{=}Y$. Applying Lemma  \ref{distributionalequalitylemma}, we obtain $\frac{h_2(\textbf{a})}{h_2(\textbf{b})}=\pm 1$. Since $h_2$ is positive, we get $h_2(\textbf{a}) = h_2(\textbf{b})$. This is a contradiction. 
\end{proof}

\begin{lemma}
\label{lemma_additivity}
Let $X$ be a random variable with strictly increasing distribution function $F_X$. Let $f(x,y)$ be a function for which an inverse with respect to $y$ exists (e.g. if $f$ is injective or strictly increasing and continuous in y, as stated by the Inverse Function Theorem \citep{Inverse_function_theorem}). Let
\begin{equation}
\label{lemma_additivity_equation}
    f(x, X) \overset{D}{=}h_1(x) + h_2(X), \,\,\,\,\,\,\,\,for\,\,all\,\,\,x\in\mathbb{R}, 
\end{equation}
for some functions $h_1, h_2$, where $h_2$ is measurable.

Then, there exist functions $\tilde{h}_1, \tilde{h}_2$, where $\tilde{h}_2$ is measurable, such that 
$$
f(x, y) \overset{}{=}\tilde{h}_1(x) + \tilde{h}_2(y), \,\,\,\,\,\,\,\,for\,\,all\,\,\,x,y\in\mathbb{R}.
$$
\end{lemma}

\begin{proof}
\label{proof of lemma_additivity}Let $g(x, y):=f(x, y) - h_1(x)$. We have that $g(x, X) \overset{D}{=}g(\tilde{x}, X)$ for all $x, \tilde{x}$. 
Let $g^{-1}(x, y)$ denote the inverse of $g(x, y)$ with respect to $y$ for a given $x$. The existence of this inverse follows directly from the existence of an inverse of $f$.  

Let $a\in\mathbb{R}$. Then 
$$ \mathbb P \left (g(x,X) \le a \right)=\mathbb P \left (X \le g^{-1}(x, a) \right)=F_X(g^{-1}(x, a)), \, \forall x\in\mathbb{R},$$where $F_X$ is the distribution function of $X$. 

Therefore
$F_X(g^{-1}(x, a)) = F_X(g^{-1}(\tilde{x}, a))$ for all $x, \tilde{x}\in\mathbb{R}$. Since $F_X$ is strictly increasing function, we obtain $g^{-1}(x, a) = g^{-1}(\tilde{x}, a)$. We showed that $g^{-1}$ does not depend on $x$. Since
$$y=g^{-1}(x, a) \Leftrightarrow g(x,y)=a,$$ we also obtain that $g(x, a) = g(\tilde{x}, a)$. Therefore, $g(x, a)$ does not depend on $x$ and we can write $g(x,y) = \tilde{h}_2(y)$ for some function $\tilde{h}_2$. We showed that 
$f(x, y) = h_1(x) + \tilde{h}_2(y)$ for all $x, y$.     
\end{proof}
\textit{Remark:} We show a periodic function $f$ such that the statement of Lemma~\ref{lemma_additivity} is not valid.

Let $Y \sim N(0,1)$. Define the continuous function $F_Y(y) = P[Y\leq y]$ for $y \in \mathbb{R}$. Note that $F_Y(Y) \sim \mbox{U}(0,1)$.
Define the continuous functions $f:\mathbb{R}^2\rightarrow\mathbb{R}$, $h_1:\mathbb{R}\rightarrow\mathbb{R}$, $h_2:\mathbb{R}\rightarrow\mathbb{R}$ by
\begin{align*}
f(x,y) &= \cos(2\pi F_Y(y) + x) \quad \forall (x,y)\in\mathbb{R}^2\\
h_1(x) &= 0,  \,\,\,\,\,\,\,\,\,
h_2(y)= \cos(2\pi F_Y(y)) \quad \forall x,y \in \mathbb{R}.
\end{align*} 
Then 
$$ f(x,Y) = \cos(2\pi F_Y(Y) + x)\overset{D}{=} \cos(2\pi F_Y(Y)) \quad \forall x \in \mathbb{R}.$$
In particular, 
$$f(x,Y) \overset{D}{=} h_1(x) + h_2(Y) \quad \forall x \in \mathbb{R}.$$
However, $f(x,y)$ does not have the form $\tilde{h}_1(x)+\tilde{h}_2(y)$.

\begin{lemma}
Let $X,Y$ be continuous random variables and $f$ is a (non-random) injective function on the support of $X$. Then, 
\begin{equation}\label{trivial_identity}
X\indep Y \iff f(X)\indep Y.
\end{equation}
\end{lemma}
\begin{proof}
This statement is trivial. 
\end{proof}

\section*{D. Appendix: Proofs}
\label{Section_proofs}

\begin{customprop}{
\ref{TheoremFidentifiabilityWithChild}}
Let $(X_0, \textbf{X})$ follow an (identifiable) restricted $\mathcal{F}_A$-model with DAG $\mathcal{G}$, such that all $\textbf{X}$ are neighbors of $X_0$ in $\mathcal{G}$.   Let $S \subseteq \{1, \dots, p\}$ contain a child of $X_0$ in $\mathcal{G}$.  Then, $S$ is not $\mathcal{F}_A$-plausible. 
\end{customprop}

\begin{proof}
 \label{proof of TheoremFidentifiabilityWithChild} 
    For a contradiction, let $S$ be $\mathcal{F}_A$-plausible. Without loss of generality, let $X_1$ be a childless child of $X_0$ such that $1\in S$ (the set of children of $X_0$ is nonempty by assumption, and one of them must be childless to avoid cycles). The idea of the proof is that we define two bivariate $\mathcal{F}_A$-models, one with $X_0\to X_1$ and one with $X_1\to X_0$, which will lead to a contradiction with the identifiability of the original restricted $\mathcal{F}_A$-model. 

Since $(X_0, \textbf{X})$ follow an $\mathcal{F}_A$-model, we can write  $X_i = f_i(\textbf{X}_{pa_i}) + \eta_i$, where $f_i$ are some measurable functions and $\eta_i$ are jointly independent, $i\in \{0, \dots, p\}$. Specifically, we have 
\begin{equation*}
    \begin{split}
        X_0 = f_0(\textbf{X}_{pa_0}) + \eta_0, \,\,\,\,\,\,\,X_1 = f_1(X_0, \textbf{X}_{pa_1\setminus\{1\}}) + \eta_1, 
    \end{split}
\end{equation*}
where $\eta_1 \indep \textbf{X}_{\{0, 2, 3, \dots, p\}}$ .Conditioning on $\textbf{X}_{\{ 2, 3, \dots, p\}} = \textbf{x}$, we obtain 
\begin{equation*}
    \begin{split}
\textbf{SCM\,1:} \,\,\,\,\,\,\,\,\,\,\,\,\,       X_0 = \tilde{\eta}_0, \,\,\,\,\,\,\,X_1 = f_1(X_0, \textbf{x}_{pa_1\setminus\{1\}}) + \eta_1, 
    \end{split}
\end{equation*}
where $\tilde{\eta}_0\sim X_0\mid\textbf{X}_{\{2, 3, \dots, p\}} = \textbf{x}$ and $\eta_1\indep X_0$.

From the fact that $S$ is $\mathcal{F}_A$-plausible, we can find a function $f$ such that $\eta_S:=X_0 - f(\textbf{X}_S)$ satisfies $\eta_S\indep \textbf{X}_S$. Hence, we can write $$
X_0 =f(X_1, \textbf{X}_{S\setminus\{1\}}) + \eta_S,$$ where $\eta_S\indep \textbf{X}_S$.  Conditioning on  $\textbf{X}_{\{2, 3, \dots, p\}} = \textbf{x}$, we obtain
$$
\textbf{SCM\,2:} \,\,\,\,\,\,\,\,\,\,\,\,\,      X_1 = \tilde{\eta}_1, \,\,\,\,\,X_0 =f(X_1, \textbf{x}_{S\setminus\{1\}}) + \eta_S,$$ where $\tilde{\eta}_1\sim X_1\mid X_{\{2, 3, \dots, p\}} = \textbf{x}$  and $\eta_S\indep X_1$. 

Notice that in both Models 1 and 2, the joint distribution of $(X_0, X_1)$ is equal to $P_{X_0, X_1\mid (X_2, \dots, X_p)=\textbf{x}}$ and hence, we were able to find two additive noise models generating the same joint distribution, where the first model follows restricted additive noise model. This is a direct contradiction with the definition of restricted additive noise model. Therefore,  $S$ is not $\mathcal{F}_F$-plausible. 
\end{proof}

\begin{customlem}{\ref{pairwise_implies_global}}
Pairwise identifiable $\mathcal{F}$-model defined in Appendix~\ref{Appendix_pairwise_identifiability} is identifiable.  
\end{customlem}

\begin{proof}\label{Proof of pairwise_implies_global}
For a contradiction, let there be two $\mathcal{F}$-models with causal graphs $\mathcal{G}\neq \mathcal{G}'$ that both generate the same joint distribution $P_{(X_0, \textbf{X})}$.   Using Proposition 29 in \cite{Peters2014}, variables $L,K\in \{X_0, \dots, X_p\}$ exist, such that 
\begin{itemize}
\item $K\to L$ in $\mathcal{G}$ and $L\to K$ in $\mathcal{G}'$,
\item $S:=\underbrace{\big\{pa_L(\mathcal{G})\setminus\{K\}\big\}}_\text{\textbf{Q}}\cup\underbrace{\big\{pa_K(\mathcal{G}')\setminus\{L\}\big\}}_\text{\textbf{R}}\subseteq \big\{nd_L(\mathcal{G}) \cap nd_K(\mathcal{G}')\setminus\{K,L\}\big\} $. 
\end{itemize}
For this $S$, choose $x_S$ according to the condition in the definition of pairwise identifiability. We use the notation $x_S=(x_q, x_r)$, where $q\in \textbf{Q}, r\in \textbf{R}$, and we define $K^\star := K\mid \{X_S=x_S\}$ and $L^\star := L\mid \{X_S=x_S\}$.  Now we use Lemma 36 and Lemma 37 from \cite{Peters2014}.  Since $K\to L$ in $\mathcal{G}$, we get $$K^\star=\tilde{\varepsilon}_{K^\star},\,\,\,\,\,\,\,\,\,\, L^\star = f_{L^\star}(K^\star, \varepsilon_L),$$ 
where $\tilde{\varepsilon}_{K^\star} = K\mid \{X_S=x_S\}$ and $\varepsilon_L\indep K^\star$. We obtained a bivariate $\mathcal{F}$-model with $K^\star\to L^\star$. However, the same holds for the other direction; from $L\to K$ in $\mathcal{G}'$, we get $$L^\star=\tilde{\varepsilon}_{L^\star},\,\,\,\,\,\,\,\,\,\, K^\star = f_{K^\star}(L^\star, \varepsilon_K),$$ 
 where $\tilde{\varepsilon}_{L^\star} =L\mid \{X_S=x_S\}$ and $\varepsilon_K\indep L^\star$. We obtained a bivariate $\mathcal{F}$-model with $L^\star\to K^\star$, which is a contradiction. 
\end{proof}

\begin{customprop}{\ref{proposition_for_pairwise_F_model_identifiability}}
Let $(X_0, \textbf{X})$ follow a pairwise identifiable $\mathcal{F}$-model with DAG $\mathcal{G}$, such that all $\textbf{X}$ are neighbors of $X_0$ in $\mathcal{G}$.   Let $S \subseteq \{1, \dots, p\}$ contain a child of $X_0$ in $\mathcal{G}$.  Then, $S$ is not $\mathcal{F}$-plausible. 
\end{customprop}

\begin{proof}\label{Proof of proposition_for_pairwise_F_model_identifiability}
    For a contradiction, let $S$ be $\mathcal{F}_A$-plausible. Without loss of generality, let $X_1$ be a childless child of $X_0$ such that $1\in S$ (the set of children of $X_0$ is nonempty by assumption, and one of them must be childless to avoid cycles). The idea of the proof is that we define two bivariate $\mathcal{F}$-models, one with $X_0\to X_1$ and one with $X_1\to X_0$, which will lead to a contradiction with the pairwise identifiability. 

  Since $(X_0, \textbf{X})$ follow an $\mathcal{F}$-model, we can write  $X_i = f_i(\textbf{X}_{pa_i(\mathcal{G})}, \varepsilon_i)$, where $f_i\in\mathcal{F}$ and $\varepsilon_i$ are jointly independent, $i\in \{0, 1, \dots, p\}$. We use the  pairwise identifiability condition. For a specific choice $(X_0, X_1)$ and $\tilde{S} = nd_1(\mathcal{G})\setminus\{0,1\} = \{2, \dots, p\}$ (the second equality holds since $X_1$ is a childless child),  $\textbf{x}_{\tilde{S}}: p_{\tilde{S}}(\textbf{x}_{\tilde{S}})>0$ exists, satisfying the condition that a bivariate $\mathcal{F}$-model defined as
\begin{equation}\label{tyuiop}
\tilde{X}_0=\tilde{\varepsilon}_0, \tilde{X}_1 = \tilde{f}_1(\tilde{X}_0, \tilde{\varepsilon}_1)
\end{equation}
is identifiable, where  $P_{\tilde{\varepsilon}_0} = P_{X_0\mid \textbf{X}_{\tilde{S}} = x_{\tilde{S}}}    $ and $\tilde{f}_1(x, \varepsilon) =f(\textbf{x}_{pa_1\setminus\{0\}}, x, \varepsilon)$, $\tilde{\varepsilon}_1\indep \tilde{\varepsilon}_0$ . 

From the fact that $S$ is $\mathcal{F}$-plausible, $f\in\mathcal{F}$ exists, such that $\varepsilon_S:=f^\leftarrow(\textbf{X}_S, X_0)$ satisfies $\varepsilon_S\indep \textbf{X}_S, \varepsilon_S\sim U(0,1)$. Hence, we can define a model  $$\tilde{\tilde{X}}_1 = \tilde{\tilde{\varepsilon}}_1 , \tilde{\tilde{X}}_0 = \tilde{\tilde{f}}(\tilde{\tilde{X}}_1, \varepsilon_S),$$ where $P_{\tilde{\tilde{\varepsilon}}_1} = P_{X_1\mid \textbf{X}_{\tilde{S}} = x_{\tilde{S}}}    $ and $\tilde{\tilde{f}}(\textbf{x}, \varepsilon) =f(\textbf{x}_{\tilde{S}}, x, \varepsilon)$. In this model, $\varepsilon_S\indep \tilde{\tilde{\varepsilon}}_1$. 

Now, note that $(\tilde{X}_0,\tilde{X}_1)\overset{D}{=}(\tilde{\tilde{X}}_0, \tilde{\tilde{X}}_1)$, since both sides are distributed as $\big[(X_0, X_1)\mid X_{\tilde{S}}\big]$.  This is a contradiction with the identifiability of (\ref{tyuiop}). Therefore,  $S$ is not $\mathcal{F}_F$-plausible.
\end{proof}

\begin{customlem}{\ref{LemmaAboutUnidentifiabilityFL}}
Let $(Y, \textbf{X})\in\mathbb{R}\times \mathbb{R}^p$ follow an $\mathcal{F}_L$-model with DAG $\mathcal{G}_0$ and $pa_Y(\mathcal{G}_0)\neq\emptyset$. Then, $|S_{\mathcal{F}_L}(Y)| \leq 1$ ($|S|$ represents the number of elements of the set $S$). Moreover, if  $a,b\in an_Y(\mathcal{G}_0)$ that are d-separated in $\mathcal{G}_0$ exist, then $S_{\mathcal{F}_L}(Y) = \emptyset$. 
\end{customlem}

\begin{proof}\label{Proof of LemmaAboutUnidentifiabilityFL}
First, we show $|S_{\mathcal{F}_L}(Y)| \leq 1$. Let $a\in an_Y(\mathcal{G}_0)\cap Source(\mathcal{G}_0)$. Such $a$ exists since $pa_Y(\mathcal{G}_0)\neq\emptyset$. We show that $S = \{a\}$ is an $\mathcal{F}_L-$plausible set. 

Denote $X_0:= Y$.  Since $\mathcal{G}_0$ is acyclic, it is possible to express recursively each variable $X_j , j= 0, \dots, p, $ as a weighted sum of the noise terms $\varepsilon_0, \dots, \varepsilon_p$ that belong to the ancestors of $X_j$. Let us write the Linear SCM with notation 
\begin{equation*}
X_i = \sum_{j\in pa_i}\beta_{j,i}X_j + \varepsilon_i  = \sum_{j\in an_i}\beta_{j\to i}\varepsilon_j,
\end{equation*}
where $\beta_{j,i}$ are non-zero constants and $\beta_{j\to i}$ is the sum of distinct weighted directed paths from node $j$ to node $i$, with a convention $\beta_{j\to j} := 1$.\footnote{To provide an example of the notation, if $X_1 = \varepsilon_1, X_2 = 2X_1 + \varepsilon_2, X_3 = 3X_1 + 4X_2+\varepsilon_3$, then $X_3 = 11\varepsilon_1 + 4\varepsilon_2 + 1\varepsilon_3 = \beta_{1\to 3}\varepsilon_1 +\beta_{2\to 3}\varepsilon_2 + \beta_{3\to 3}\varepsilon_1 $.}

Using this notation, note that 
$$X_0 = \sum_{j\in an_0}\beta_{j\to i}\varepsilon_j = \beta_{a\to 0}\varepsilon_a + \sum_{j\in an_0\setminus \{a\}}\beta_{j\to i}\varepsilon_j =\beta_{a\to 0} X_a + \sum_{j\in an_0\setminus \{a\}}\beta_{j\to i}\varepsilon_j ,$$ where $X_a\indep \sum_{j\in an_i\setminus \{a\}}\beta_{j\to i}\varepsilon_j$ since $a\in Source(\mathcal{G}_0)$. Hence, $Y - \beta_{a\to 0} X_a\indep X_a$, which is almost the definition of $\mathcal{F}_L$-plausibility of set $S = \{a\}$. More rigorously, for  $S = \{a\}$, we can find $f\in\mathcal{F}_L$ such that  $f_Y^{\leftarrow}({X}_{S}, Y)\indep X_S$ and $f_Y^{\leftarrow}({X}_{S}, Y)\sim U(0,1)$. This function can be defined as 
$$
f(x, \varepsilon) = \beta_{a\to 0}x + g^{-1}(\varepsilon), \,\,\,x\in\mathbb{R}, \varepsilon\in (0,1),
$$
where $g$ is the distribution function of $(Y - \beta_{a\to 0}X_S)$. This function obviously satisfies $f\in\mathcal{F}_L$. Moreover, since $f_Y^{\leftarrow}(X_{S}, Y) = g(Y - \beta_{a\to 0}X_S)$, it holds that $f_Y^{\leftarrow}({X}_{S}, Y)\indep X_S$ and $f_Y^{\leftarrow}({X}_{S}, Y)\sim U(0,1)$, which is what we wanted to show. Hence,  $|S_{\mathcal{F}_L}(Y)|\leq 1$, since $S_{\mathcal{F}_L}(Y)\subseteq S = \{a\}$. 

Now, let $a,b\in an_Y(\mathcal{G}_0)$ that are d-separated in $\mathcal{G}_0$. Let $a', b'\in \mathcal{G}_0$ such that $a'\in \big\{an_a(\mathcal{G}_0)\cup \{a\}\big \}\cap Source(\mathcal{G}_0)$,   $b'\in \big\{an_b(\mathcal{G}_0)\cup \{b\}\big \}\cap Source(\mathcal{G}_0)$. They are well defined since the sets $an_a(\mathcal{G}_0)\cup \{a\}$,  $an_b(\mathcal{G}_0)\cup \{b\}$  must contain some source node. Since $a,b$ are d-separated,   $ \big\{an_a(\mathcal{G}_0)\cup \{a\}\big \}$ and $\big\{an_b(\mathcal{G}_0)\cup \{b\}\big \}$ are disjoint sets, $a'\neq b'$ (they are even d-separated in $\mathcal{G}_0$). 

Using the same argument as in the first part of the proof, since $a'\in an_Y(\mathcal{G}_0)\cap Source(\mathcal{G}_0)$, it holds that  $S = \{a\}$ is an $\mathcal{F}_L-$plausible set.  $S = \{b\}$ is also an $\mathcal{F}_L-$plausible set since $b'\in an_Y(\mathcal{G}_0)\cap Source(\mathcal{G}_0)$. Together, $S_{\mathcal{F}_L}(Y)\subseteq \{a\}$ and $S_{\mathcal{F}_L}(Y)\subseteq \{b\}$. We showed that  $S_{\mathcal{F}_L}(Y) = \emptyset$. 
\end{proof}

\begin{customlem}{\ref{lemma158}}
 Let $\mathcal{F}\subseteq\mathcal{I}_m$. Let $(X_0, \textbf{X})\in\mathbb{R}\times \mathbb{R}^p$ follow an $\mathcal{F}$-model with DAG $\mathcal{G}_0$ and $pa_{X_0}(\mathcal{G}_0)\neq \emptyset$. Let  $S\subseteq \{1, \dots, p\}$ be a non-empty set. If $(X_0, \textbf{X})$ is marginalizable to $S\cup\{0\}$, then $S_{\mathcal{F}}(X_0)\subseteq S$. 
\end{customlem}
\begin{proof}
 \label{Proof of lemma158}
 Since  $(X_0, \textbf{X})$ is marginalizable to $S\cup\{0\}$,  $(X_0, \textbf{X}_S)$ follows an $\mathcal{F}$-model. Therefore, $f_0\in\mathcal{F}$ exists, such that $X_0 = f_0(X_{\tilde{S}}, \varepsilon_0)$ for some $\tilde{S}\subseteq S$, $\varepsilon_0\indep X_{\tilde{S}}$, $\varepsilon_0\sim U(0,1)$. 
In other words, $ f_0^{\leftarrow}(X_{\tilde{S}}, X_0)\indep X_{\tilde{S}}$,  $f_0^{\leftarrow}(X_{\tilde{S}}, X_0)\sim U(0,1)$, which is exactly the definition of $\mathcal{F}$-plausibility. Hence, $\tilde{S}$ is $\mathcal{F}$-plausible and consequently,  $S_{\mathcal{F}}(X_0)\subseteq \tilde{S}\subseteq S$. 
 \end{proof}

\begin{customprop}{\ref{Support_proposition}}
Let $(Y, \textbf{X})\in\mathbb{R}\times \mathbb{R}^p$ follow an SCM with DAG $\mathcal{G}_0$. Let $S\subseteq\{1, \dots, p\}$ be a non-empty set.

\begin{itemize}
    \item (Additive case) Let  $\underline{\Psi}: \mathbb{R}^{\mid S\mid}\to \mathbb{R}$ be a finite lower support of $ (Y\mid \textbf{X}_S=\textbf{x}_S)$. 
If
\begin{equation} \tag{\ref{eq9987_additive}}
Y - \underline{\Psi}(\textbf{X}_S)\not\indep \textbf{X}_S,
\end{equation}
then $S$ is not $\mathcal{F}_{A}$-plausible. 

 \item (Location-scale case) Let  $\underline{\Psi},\overline{\Psi}: \mathbb{R}^{\mid S\mid}\to \mathbb{R}$ be real functions such that
\begin{equation*} 
supp(Y\mid \textbf{X}_S=\textbf{x}) = \big(\underline{\Psi} (\textbf{x}),\overline{\Psi}(\textbf{x})\big), \,\,\,\,\,\,\, \forall \textbf{x}\in supp(\textbf{X}_S).
\end{equation*}
If
\begin{equation}\tag{\ref{eq9987}}
\frac{Y - \underline{\Psi}(\textbf{X}_S)}{\overline{\Psi}(\textbf{X}_S) - \underline{\Psi}(\textbf{X}_S)}\not\indep \textbf{X}_S,
\end{equation}
then $S$ is not $\mathcal{F}_{LS}$-plausible. 
\end{itemize}
\end{customprop}

\begin{proof}\label{Proof of Support_proposition}
\textbf{First bullet-point: }
For a contradiction, let $S$ be  $\mathcal{F}_{A}$-plausible. Hence, $f\in\mathcal{F}_{A}$ exists such that 
\begin{equation}\label{eq74_additive}
f^{\leftarrow}(\textbf{X}_S, Y)\indep\textbf{ X}_S.
\end{equation}
Since $f\in\mathcal{F}_{A}$, we can write $f^{\leftarrow}(\textbf{x},y) = q\big(y - \mu(\textbf{x})\big)$  for some function $\mu(\cdot)$ and for some  distribution function $q(\cdot)$. Using this notation, (\ref{eq74_additive}) is equivalent to
\begin{equation} \label{eq989_additive}
Y - {\mu}(\textbf{X}_S)\indep \textbf{X}_S.
\end{equation}
Denote $W_{\textbf{x}}:= (Y\mid \textbf{X}_S=\textbf{x})$. From (\ref{eq989_additive}), we get that for all $\textbf{x},\textbf{y}$ in the support of $\textbf{X}_S$, it must hold that
\begin{equation}\label{eq7285_additive}
W_\textbf{x} - {\mu}(\textbf{x})\overset{D}{=}W_\textbf{y} - {\mu}(\textbf{y}).
\end{equation}
Hence, supports must also match, i.e., (\ref{eq7285_additive}) implies
\begin{align*}
\ \underline{\Psi} (\textbf{x}) - {\mu}(\textbf{x})&\overset{}{=} \underline{\Psi} (\textbf{y}) - {\mu}(\textbf{y}),
\end{align*}
for all $\textbf{x},\textbf{y}$ in the support of $\textbf{X}_S$. Solving for $\mu$  gives us 
\begin{align*}
{\mu}(\textbf{x})&\overset{}{=}c_1+ \underline{\Psi} (\textbf{x}) , 
\end{align*}
where $c_1 \in  \mathbb{R}$ is some constant. Plugging this into (\ref{eq989_additive}) gives us a contradiction with (\ref{eq9987_additive}). 

\textbf{Second bullet-point:} For a contradiction, let $S$ be  $\mathcal{F}_{LS}$-plausible. Hence, $f\in\mathcal{F}_{LS}$ exists such that 
\begin{equation}\label{eq74}
f^{\leftarrow}(\textbf{X}_S, Y)\indep\textbf{ X}_S.
\end{equation}
Since $f\in\mathcal{F}_{LS}$, we can write $f^{\leftarrow}(\textbf{x},y) = q\big(\frac{y - \mu(\textbf{x})}{\sigma(\textbf{x})}\big)$  for some functions $\mu(\cdot), \sigma(\cdot)>0$ and for some (continuous) distribution function $q(\cdot)$. Using this notation, (\ref{eq74}) is equivalent to
\begin{equation} \label{eq989}
\frac{Y - {\mu}(\textbf{X}_S)}{{\sigma}(\textbf{X}_S)}\indep \textbf{X}_S.
\end{equation}
Denote $W_{\textbf{x}}:= (Y\mid \textbf{X}_S=\textbf{x})$. From (\ref{eq989}), we get that for all $\textbf{x},\textbf{y}$ in the support of $\textbf{X}_S$, it must hold that
\begin{equation}\label{eq7285}
\frac{W_\textbf{x} - {\mu}(\textbf{x})}{{\sigma}(\textbf{x})}\overset{D}{=}\frac{W_\textbf{y} - {\mu}(\textbf{y})}{{\sigma}(\textbf{y})}.
\end{equation}
Hence, supports must also match, i.e., (\ref{eq7285}) implies
\begin{align*}
\frac{ \underline{\Psi} (\textbf{x}) - {\mu}(\textbf{x})}{{\sigma}(\textbf{x})}&\overset{}{=}\frac{ \underline{\Psi} (\textbf{y}) - {\mu}(\textbf{y})}{{\sigma}(\textbf{y})}, \,\,\,\,\,\,\,\,\,\,\,\,\,\,\,\,\,\,\,\,\,\,\,\,\,\,
\frac{ \overline{\Psi}(\textbf{x}) - {\mu}(\textbf{x})}{{\sigma}(\textbf{x})}\overset{}{=}\frac{ \overline{\Psi}(\textbf{y}) - {\mu}(\textbf{y})}{{\sigma}(\textbf{y})},
\end{align*}
for all $\textbf{x},\textbf{y}$ in the support of $\textbf{X}_S$. Solving for $\mu, \sigma$  gives us 
\begin{align*}
{\mu}(\textbf{x})&\overset{}{=}c_1+ \underline{\Psi} (\textbf{x}) , \,\,\,\,\,\,\,\,\,\,\,
\sigma(\textbf{x})\overset{}{=}c_2\cdot [\overline{\Psi} (\textbf{x})-\underline{\Psi}(\textbf{x})],
\end{align*}
where $c_1 \in  \mathbb{R}, c_2 \in \mathbb{R}_{+}$ are some constants. Plugging this into (\ref{eq989}) gives us a contradiction with (\ref{eq9987}). 
\end{proof}

\begin{customprop}{\ref{LemmaOLocationScaleinseparabilite}}
Let \((Y, \mathbf{X}) \in \mathbb{R} \times \mathbb{R}^p\) is continuous and satisfy \eqref{SCM_for_Y} with \(pa_Y \neq \emptyset\) and $\mathcal{F} = \mathcal{F}_{LS}$. Then,  $S\subsetneq pa_Y$ is not $\mathcal{F}_{LS}$-plausible if $\textbf{X}_{pa_Y}$ has independent components and $f_Y$  have the form $$f_Y(\textbf{x}, \varepsilon)=\mu(\textbf{x}) + \sigma(\textbf{x})\varepsilon, $$where $\theta(\textbf{x}) = \big(\mu(\textbf{x}), \sigma(\textbf{x})\big)^\top$ is additive in both components, that is,  
$\mu(\textbf{x}) = h_{1, \mu}(x_1)+\dots + h_{k, \mu}(x_k)$ and 
$\sigma(\textbf{x}) =h_{1, \sigma}(x_1)+\dots + h_{k, \sigma}(x_k)$ for some continuous non-constant non-zero functions $h_{i,\cdot}$, where we also assume $h_{i,\sigma}>0$, $i=1, \dots, k$.  
\end{customprop}
\begin{proof}
\label{Proof of LemmaOLocationScaleinseparabilite}
For a contradiction, consider that $S$ is $\mathcal{F}_{LS}$-plausible. Without loss of generality, let $S=\{1, \dots, s\}$ for $s<k=|pa_Y|$. Then, $g\in\mathcal{F}_{LS}$ exist such that  $g^{\leftarrow}\big(\textbf{X}_S, Y\big)\indep \textbf{X}_S$. Since $g\in\mathcal{F}_{LS}$, we can write $g(\textbf{x}_S, e) = \mu_g(\textbf{x}_S) + \sigma_g(\textbf{x}_S)q^{-1}(e)$ for some function $\theta_g=(\mu_g, \sigma_g)$ that is minimal almost surely and hence non-constant in neither of the arguments. Inverse of such a function is in the form $g^{\leftarrow}(\textbf{x}_S, e) = q(\frac{e - \mu_g(\textbf{x}_S)}{\sigma_g(\textbf{x}_S)})$. 

Hence, simply rewriting   
$$\textbf{X}_S\indep g^{\leftarrow}\big(\textbf{X}_S, Y\big)
= q\bigg(\frac{Y - \mu_g(\textbf{X}_S)}{\sigma_g(\textbf{X}_S)}\bigg),$$
and using \eqref{trivial_identity} and $Y = \mu(\textbf{X}_{pa_Y}) + \sigma(\textbf{X}_{pa_Y})\varepsilon_Y$ we get 
\begin{equation}\label{dfrgeTRI}
\textbf{X}_S\indep \frac{\mu(\textbf{X}_{pa_Y}) + \sigma(\textbf{X}_{pa_Y})\varepsilon_Y - \mu_g(\textbf{X}_S)}{\sigma_g(\textbf{X}_S)}. 
\end{equation}
Equation (\ref{dfrgeTRI}) can be equivalently rewritten into
\begin{equation}\label{erty}
\textbf{X}_S\indep f_1(\textbf{X}_S)  + f_2(\textbf{X}_S)h(\textbf{X}_{S^c}, \varepsilon_Y),
\end{equation}
where $ S^c=pa_Y\setminus S$, and 
$$f_1(\textbf{x}) =\frac{ h_{1, \mu}(x_1)+\dots + h_{s, \mu}(x_s)- \mu_g(\textbf{x})}{\sigma_g(\textbf{x})},\,\,\,\,\,\,\,\,\,\,f_2(\textbf{x}) =  \frac{ h_{1, \mu}(x_1)+\dots + h_{s, \mu}(x_s)}{\sigma_g(\textbf{x})},$$
$$
h(\textbf{x}, \varepsilon) = h_{s+1, \mu}(x_{s+1})+\dots + h_{k, \mu}(x_k) + [h_{s+1, \sigma}(x_{s+1})+\dots + h_{k, \sigma}(x_k)]\varepsilon.
$$
However, independence (\ref{erty}) is a contradiction with Lemma \ref{CoolLemma} part 4. 
\end{proof}

\begin{customprop}{\ref{LemmaOParetoinseparabilite}}
 Consider $(Y, \textbf{X})\in\mathbb{R}\times \mathbb{R}^p$ satisfying \eqref{SCM_for_Y} with  $pa_Y\neq \emptyset$ and $\mathcal{F} = \mathcal{F}_{F}$. where $F$ be a distribution function whose parameter acts multiplicatively. 
 Let $\textbf{X}_{pa_Y}$ be a continuous random vector with full support and independent components. 
\begin{itemize}
\item Consider $f_Y\in\mathcal{F}_F$ in the form $f_Y(\textbf{x}, \varepsilon)=F^{-1}\big(\varepsilon, \theta(\textbf{x})\big)$ with additive function $\theta(x_1, \dots, x_k) = h_1(x_1)+\dots + h_k(x_k)$, where $h_i$ are continuous non-constant real functions. Then, then every $S\subsetneq pa_Y$ is not $\mathcal{F}_{F}$-plausible. 
\item Consider $f_Y\in\mathcal{F}_F$ in the form $f_Y(\textbf{x}, \varepsilon)=F^{-1}\big(\varepsilon, \theta(\textbf{x})\big)$ with multiplicative function   $\theta(x_1, \dots, x_k) = h_1(\textbf{x}_S)\cdot h_2(\textbf{x}_{\{1, \dots, k\}\setminus S})$ for some $S\subsetneq \{1, \dots, k\}$, where $h_1, h_2$ are continuous non-constant non-zero real functions. Then, $S_{\mathcal{F}_F}(Y)=\emptyset$.
\end{itemize}
\end{customprop}
\begin{proof}\label{Proof of LemmaOParetoinseparabilite}
\textbf{The first bullet-point}: For a contradiction, consider that $S=\{1, \dots, s\}\subset \{1, \dots, k\}$ is $\mathcal{F}_{F}$-plausible. Then for almost all $z\in(0,1)$, there exist $g\in\mathcal{F}_F$ such that  $g^{\leftarrow}\big(\textbf{X}_S, f(\textbf{X}, z)\big)\indep \textbf{X}_S$. Since $g\in\mathcal{F}_F$, we can write $g^{\leftarrow}(\textbf{x}_S, \cdot) = F\big(\cdot, \theta_g(\textbf{x}_S)\big)$ for some non-constant function $\theta_g$. Hence, 
$$\textbf{X}_S\indep g^{\leftarrow}\big(\textbf{X}_S, f(\textbf{X}, z)\big)
= F[F^{-1}\big(z, \theta(\textbf{X})\big), \theta_g(\textbf{X}_S)]
= f_1[z, f_2\big(\theta_g(\textbf{X}_S)\big) \cdot\theta(\textbf{X})].$$ 
We use identity (\ref{trivial_identity}). Since $f_1$ is invertible, we obtain 
\begin{equation}\label{dfrge}
\textbf{X}_S\indep f_2\big(\theta_g(\textbf{X}_S)\big) \cdot\theta(\textbf{X}).
\end{equation}
Define $\tilde{\theta}_g(\textbf{X}_S):=f_2\big(\theta_g(\textbf{X}_S)\big)$. Finally, since $\theta(\textbf{X})$ is an additive function from the assumptions, (\ref{dfrge}) is equivalent to
$$
\tilde{\theta}_g(\textbf{X}_S)[h_1(X_1) + \dots + h_k(X_k)]\indep \textbf{X}_S. 
$$
However, that is a contradiction with Lemma \ref{CoolLemma} part 2. 

\textbf{The second bullet-point}:  
We show that $S$ is $\mathcal{F}_F$-plausible set by finding an appropriate function $g\in\mathcal{F}_F$ such that $g^{\leftarrow}\big(\textbf{X}_S, f_Y(\textbf{X}, \varepsilon_Y)\big)\indep \textbf{X}_S$. Since it must hold that  $g\in\mathcal{F}_F$, we write  $g^{\leftarrow}(\textbf{x}_S, \cdot) = F\big(\cdot, \theta_g(\textbf{x}_S)\big)$ for some $\theta_g$. 

Rewrite 
 $$g^{\leftarrow}\big(\textbf{X}_S, f(\textbf{X}, \varepsilon_Y)\big) = F[F^{-1}\big(\varepsilon_Y, \theta(\textbf{X})\big)     ,\theta_g(\textbf{X}_S)]
 = f_1[\varepsilon_Y, f_2\big(\theta_g(\textbf{X}_S)\big) \cdot\theta(\textbf{X})],$$
where $f_1, f_2$ are from (\ref{postMultiplDefinition}). 
We choose $\theta_g$ such that $f_2\big(\theta_g(\textbf{x}_S)\big) = \frac{1}{h_1(\textbf{x}_S)}$. Obviously, $g\in\mathcal{F}_F$. Then, by extending $\theta$ to its multiplicative form, we get 
$$ f_1[\varepsilon_Y, f_2\big(\theta_g(\textbf{X}_S)\big) \cdot\theta(\textbf{X})] =  f_1[\varepsilon_Y,  h_2(\textbf{X}_{\{1, \dots, k\}\setminus S})]\indep \textbf{X}_S.$$
Together, we found $g\in\mathcal{F}_F$ defined by $g^{\leftarrow}(\textbf{x}_S, \cdot) = F\big(\cdot, f_2^{-1}(\frac{1}{h_1(\textbf{x}_S)})\big)$ that satisfy   $g^{\leftarrow}\big(\textbf{X}_S, f(\textbf{X}, \varepsilon_Y)\big)\indep \textbf{X}_S$. Hence, $S$ is $\mathcal{F}_F$-plausible. The analogous argument can be given for the set $\{1, \dots, k\}\setminus S$. Hence,  $S_{\mathcal{F}_F}(Y)\subseteq S\cap( \{1, \dots, k\}\setminus S) = \emptyset$.
\end{proof}

\begin{customthm}{\ref{Theorem_in_section2}}
Let \((Y, \mathbf{X}) \in \mathbb{R} \times \mathbb{R}^p\) be a continuous random vector that satisfy \eqref{SCM_for_Y}. Suppose  \(\mathcal{F} = \mathcal{F}_A\) and consider a non-empty set \(S \subsetneq pa_Y\). 

\begin{itemize}
 \item (Independent case) Assume that \(\mathbf{X}_{pa_Y}\) has independent components and \(f_Y\) is an injective function. Then, the set \(S\) is \(\mathcal{F}_A\)-plausible if and only if \(f_Y\) can be decomposed as follows:  
    \begin{equation}\tag{\ref{decomposition_condition_independent}}
        f_Y(\mathbf{x}, e) = h_1(\mathbf{x}_S) + h_2(\mathbf{x}_{pa_Y \setminus S}) + q^{-1}(e), \quad\quad \forall
        \mathbf{x} \in \mathbb{R}^{|pa_Y|}, \, e \in (0, 1),
    \end{equation}
    where \(h_1, h_2\) are measurable functions, \(q^{-1}\) is a quantile function, and \(pa_Y \setminus S = \{i \in pa_Y : i \not\in S\}\).

    \item (General case) The set \(S\) is \(\mathcal{F}_A\)-plausible if and only if the function \(f_Y\) can be expressed as:
    \begin{equation}\tag{\ref{decomposition_condition}}
        f_Y(\mathbf{x}, e) = h_1(\mathbf{x}_S) + h_2(\mathbf{x}) + q^{-1}(e), \quad\quad \forall
        \mathbf{x} \in \mathbb{R}^{|pa_Y|}, \, e \in (0, 1),
    \end{equation}
    for some measurable function \(h_1\), quantile function \(q^{-1}\), and a function \(h_2\) such that 
    \[
    h_2 \in \mathcal{H}_{\mathbf{X}_{pa_Y}}(S) := \{f : \mathbb{R}^{|pa_Y|} \to \mathbb{R} \mid f(\mathbf{X}_{pa_Y}) \indep \mathbf{X}_S\}.
    \]

    \item As a consequence, \(S_{\mathcal{F}_A}(Y) = pa_Y\) if and only if:  
    \begin{enumerate}
        \item \(f_Y\) cannot be expressed in the form of \eqref{decomposition_condition} for any \(S \subsetneq pa_Y\), and  
        \item every set \(S\) that is neither a subset nor a superset of \(pa_Y\) (i.e., \(pa_Y \not\subseteq S \not\subseteq pa_Y\)) is not \(\mathcal{F}_A\)-plausible (e.g., under the assumptions of Proposition~\ref{TheoremFidentifiabilityWithChild}).  
    \end{enumerate}
\end{itemize}
\end{customthm}

\begin{proof}
\label{Proof of Theorem_in_section2}

    \textbf{The second bullet-point}:   ''$\impliedby$'' if $f_Y$ has a form \eqref{decomposition_condition}, then \( S \) is an \( \mathcal{F}_A \)-plausible set since we can find \( f \in \mathcal{F}_A \) such that $f^{\leftarrow}(X_S, Y) \indep X_S$. Such a choice is $f(\textbf{x}_S, e) = h_1(\textbf{x}_s) + q^{-1}(e)$, since $Y - f(\textbf{X}_S, \varepsilon_Y) =h_1(\textbf{X}_S)+  h_2(\textbf{X}_{pa_Y}) + q^{-1}(\varepsilon_Y) - h_1(\textbf{X}_S) - q^{-1}(\varepsilon_Y) =  h_2(\textbf{X}_{pa_Y})\indep \textbf{X}_S$. 
    
    ''$\implies$'':  Assume that \( S \) is an \( \mathcal{F}_A \)-plausible set; hence there exists \( f \in \mathcal{F}_A \) such that $f^{\leftarrow}(\textbf{X}_S, Y) \indep \textbf{X}_S$. We use the following notation: \( f(\textbf{x}, e) = \mu(\textbf{x}) + \tilde{q}^{-1}(e) \) for \( \textbf{x} \in \mathbb{R}^{|S|} \) and \( e \in (0,1) \), where \( \mu \) is some function and \( \tilde{q}^{-1} \) is a quantile function. Additive functions have an inverse in the form \( f^{\leftarrow}(x, y) = \tilde{q}(y - \mu(x)) \) for \( x \in \mathbb{R}^{|S|} \) and \( y \in \mathbb{R} \) (see the discussion in Appendix~\ref{Appendix_A.1.}). Notice that due to an assumption of continuity of $(Y, \textbf{X})$, $ \tilde{q}^{-1}$ must be injective on the support of $Y$ and we can use identity the \eqref{trivial_identity}. We therefore have:
\begin{equation*}
    \begin{split}
\text{S is } \mathcal{F}_A\text{-plausible}& \iff f^{\leftarrow}(\textbf{X}_S, Y) \indep\textbf{ X}_S \iff Y - \mu(\textbf{X}_S) \indep \textbf{X}_S \\&
\iff    f_0(\textbf{X}_{pa_Y}) + q^{-1}(\varepsilon_Y) - \mu(\textbf{X}_S) \indep\textbf{ X}_S        
\\&
\iff    f_0(\textbf{X}_{pa_Y})  - \mu(\textbf{X}_S) \indep\textbf{ X}_S,        
    \end{split}
\end{equation*}
where we used a notation \( Y = f_0(X_1, X_2) + q^{-1}(\varepsilon_Y) \), where \( \varepsilon_Y \indep \textbf{X}_{pa_Y} \). Therefore,  we directly obtain $f_0(x_1, x_2) - \mu(x_1)\in \mathcal{H}_{\textbf{X}}(S) $. Defining $h_1 = \mu, h_2 = f_0 - \mu$ we obtain precisely \eqref{decomposition_condition}.

    \textbf{The first bullet-point}:  ''$\impliedby$'' if $f_Y$ has a form \eqref{decomposition_condition_independent}, then \( S \) is an \( \mathcal{F}_A \)-plausible set since for $f(\textbf{x}_S, e) = h_1(\textbf{x}_S) + q^{-1}(e)$ trivially holds  $f^{\leftarrow}(\textbf{X}_S, Y)\indep \textbf{X}_S$. 
    
    ''$\implies$'':  Consider that $S$ is an  $\mathcal{F}_A$-plausible set; hence there exists \( f \in \mathcal{F}_A \) such that $f^{\leftarrow}(\textbf{X}_S, Y) \indep \textbf{X}_S$.  Let us write $$Y = f_0(X_1, \dots, X_{|pa_Y|}) + q^{-1}(\varepsilon_Y), \,\,\,\,\,\,\,\,\varepsilon_Y\indep \textbf{X}_{pa_Y},$$ for some injective function $f_0$ and injective quantile function $q^{-1}$, and let us write $$f(\textbf{x}, e) = \mu(\textbf{x}) + \tilde{q}^{-1}(e), \,\,\textbf{x}\in\mathbb{R}^{|S|}, e\in (0,1),$$ for some measurable function $\mu$ and quantile function $\tilde{q}^{-1}$. Additive functions have an inverse in a form $f^{\leftarrow}(\textbf{x}, y) = \tilde{q}\big(y-\mu(\textbf{x})\big),\,\,\textbf{x}\in\mathbb{R}^{|S|}, y\in\mathbb{R}$ (see discussion in Appendix~\ref{Appendix_A.1.}).   Using Definition~\ref{Definition1} and identity \eqref{trivial_identity}, we have $Y- \mu(\textbf{X}_S)\indep \textbf{X}_S$. 

Hence, we have
\begin{align*}
    Y- \mu(\textbf{X}_S)&\indep \textbf{X}_S 
    \\ f_0(X_1, \dots, X_{|pa_Y|}) + q^{-1}(\varepsilon_Y) - \mu(\textbf{X}_S) &\indep \textbf{X}_S 
    \\f_0(X_1, \dots, X_{|pa_Y|}) - \mu(\textbf{X}_S) &\indep \textbf{X}_S. 
\end{align*}
Lemma \ref{CoolLemma} part 1 (in particular, the negation of that statement) directly shows that this implies the form~\eqref{decomposition_condition_independent}.  

\textbf{The third bullet point} follows directly from Proposition~\ref{TheoremFidentifiabilityWithChild} and the previous bullet point. Specifically, \(S_{\mathcal{F}}(Y) = pa_Y \iff\) every set \(S\) with \(S \not\supseteq pa_Y\) is not \(\mathcal{F}\)-plausible. This is equivalent to stating that any set \(S\) such that either \(S \subsetneq pa_Y\) or \(pa_Y \not\subseteq S \not\supseteq pa_Y\) is not \(\mathcal{F}\)-plausible. The case \(S \subsetneq pa_Y\) follows from the first assumption, while \(pa_Y \not\subseteq S \not\supseteq pa_Y\) follows from the second.
\end{proof}

\begin{customthm}{\ref{theorem_consistency_ISD}}
Let $(Y, \mathbf{X})$ satisfy \eqref{SCM_for_Y} with $pa_Y \neq \emptyset$.  Assume that the estimator $\hat{S}_{\mathcal{F}}(Y)$ is constructed as described above, using $\hat{f}_{pa_Y} = f_Y$ and valid tests $H_{0, S}^I, H_{0, S}^S, H_{0, S}^D$  for all $S\subseteq \{1, \dots, p\}$ at level $\alpha$ in a sense that for all $S$, $\sup_{P: H_{0, S}^\cdot \text{is true}}P(H_{0, S}^\cdot \text{ is rejected})\leq \alpha$ for all $\cdot \in \{S, I, D\}$. Then
\begin{equation}\label{Theorem2_3alpha}
    P(\hat{S}_{\mathcal{F}}(Y) \subseteq pa_Y) \geq 1-3\alpha. 
\end{equation}

Furthermore, suppose ${S}_{\mathcal{F}}(Y) = pa_Y$, and assume that all tests have non-zero power, i.e., $lim_{n\to\infty}P(H_{0, S}^\cdot \text{ is rejected}\mid  H_{0, S}^\cdot \text{ is false})=1$ for all $\cdot \in \{S, I, D\}$ and all $S\not\supseteq pa_Y$. Then, there exists an integer $n_0$ such that for all $n \geq n_0$, it holds that
\begin{equation}\label{Theorem2_3alpha_eq2}
    P(\hat{S}_{\mathcal{F}}(Y) = pa_Y) \geq 1-3\alpha. 
\end{equation}
\end{customthm}
\begin{proof}
\label{Proof of theorem_consistency_ISD}
\eqref{Theorem2_3alpha} follows directly from the following computations:
\begin{align*}
    & P(\hat{S}_{\mathcal{F}}(Y) \subseteq pa_Y) \geq  P(H_{0, pa_Y}(\mathcal{F}) \text{ is not rejected}) =  P(H_{0, pa_Y}^I, H_{0, pa_Y}^S, H_{0, pa_Y}^D \text{ are not rejected})\\& \geq \prod_{\cdot \in \{I, S, D\} } P(H_{0, pa_Y}^\cdot\text{ is not rejected}) \geq (1-\alpha)^3 > 1-3\alpha \,\,\,\,\,\,\,\,for\,\,\alpha\in(0,1). 
\end{align*} 
As for \eqref{Theorem2_3alpha_eq2}, note that the identifiability of the parents  ${S}_{\mathcal{F}}(Y) =pa_Y$ implies that for any $S\not\supseteq pa_Y$, and any function $\hat{f}_S$, either $\hat{f}_S\not\in\mathcal{F}$, or $\hat{\varepsilon}_S \not\indep \textbf{X}_S$, or $\hat{\varepsilon}_S \not\sim U(0,1)$. In other words, either $H_{0,S}^S$, $H_{0, S}^I$ or $H_{0, S}^D$ is not true. Therefore, there exist $n_0$ such that  $H_{0, S}^\cdot$ will be rejected with large probability, where $\cdot \in \{I, S, D\}$ correspond to the non-true hypothesis. Therefore, 
\begin{align*}
        &\lim_{n\to\infty}P(\hat{S}_{\mathcal{F}}(Y) \neq pa_Y) \leq \lim_{n\to\infty}P(H_{0, pa_Y}(\mathcal{F}) \text{ is rejected, or }\exists S \not\supseteq pa_Y: H_{0, S}(\mathcal{F}) \text{ is not rejected})\\&
        \leq  1-(1-\alpha)^3 +    \lim_{n\to\infty} P(\exists S \not\supseteq pa_Y: H_{0, S}(\mathcal{F}) \text{ is not rejected})
        \\& =   1-(1-\alpha)^3 +    \lim_{n\to\infty} P(\exists S \not\supseteq pa_Y: H_{0, S}^I, H_{0, S}^S, H_{0, S}^D \text{ are not rejected})\\&
\leq 1-(1-\alpha)^3  +  \lim_{n\to\infty} P(\exists S \not\supseteq pa_Y: H_{0, S}^\cdot \text{ is not rejected}\mid  H_{0, S}^\cdot \text{ is false})\\&= 1-(1-\alpha)^3 + 0 <3\alpha.      
\end{align*}
\end{proof}

\begin{customlem}{\ref{lemma_hidden_confounder}}
 Consider $(Y, \textbf{X})\in\mathbb{R}\times \mathbb{R}^p$  satisfy \eqref{SCM_for_Y} with  $\mathcal{F} = \mathcal{F}_A$. Consider $\emptyset\neq hid\subset pa_Y$. Let $S\subseteq pa_Y \cap obs$ and $\tilde{S}:=(pa_Y\cap obs)\setminus S$ such that
 $(\textbf{X}^{hid}, \textbf{X}_S)\indep \textbf{X}_{\tilde{S}}$ (one can consider that $\textbf{X}^{hid}$ cause $\textbf{X}_S$ and $Y$, but not $\textbf{X}_{\tilde{S}}$). 
 
 If $f_Y$ has a form \begin{equation*}
        f_Y(\textbf{x}, e) = h_1(\textbf{X}^{hid}, \textbf{X}_S) + h_2(\textbf{X}_{\tilde{S}}) + q^{-1}(e), \,\,\,\,\,\textbf{x}\in\mathbb{R}^{|pa_Y|}, e\in(0,1),
     \end{equation*}
for some continuous non-constant real functions  $h_1, h_2$ and a quantile function $q^{-1}$. Then, $S_{\mathcal{F}_A}(Y) \subseteq \tilde{S}\subset pa_Y$.
\end{customlem}
\begin{proof}
\label{Proof of lemma_hidden_confounder}
The set $\tilde{S}$ is $\mathcal{F}_A$-plausible since  $Y - h_2(\textbf{X}_{\tilde{S}}) =h_1(\textbf{X}^{hid}, \textbf{X}_S) + q^{-1}(\varepsilon_Y)  \indep \textbf{X}_{\tilde{S}}$. Therefore $S_{\mathcal{F}_A}(Y)\subseteq\tilde{S}$. 
\end{proof} 

\begin{customprop}{\ref{Proposition_consistency}}
Consider $\mathcal{F} = \mathcal{F}_A$ and 
let $(Y, \textbf{X})\in\mathbb{R}\times \mathbb{R}^p$ follow an SCM with DAG $\mathcal{G}_0$ satisfying \eqref{SCM_for_Y}. Assume that every $S \neq pa_Y$ is not $\mathcal{F}$-plausible. 
Then,   

\begin{equation}
  \lim_{n\to\infty} \mathbb{P}(\widehat{pa}_Y  \neq pa_Y) = 0,
\end{equation}   
where $n$ is the size of the random sample and $\widehat{pa}_Y$ is our score-based estimate from Section~\ref{Section_algorithm2} with $\lambda_1, \lambda_2>0, \lambda_3 = 0$, suitable estimation procedure, and HSIC independence measure. 
\end{customprop}
\begin{proof}
\label{Proof of Proposition_consistency}
This result is a simple consequence of Theorem 20 in \cite{reviewANMMooij}. We use the same notation. For a rigorous definition of $HSIC$ and $\widehat{HSIC}$, see Appendix A.1 in \cite{reviewANMMooij}.

We show that $score(S)> score(pa_Y)$ as $n\to\infty$ for any $S\neq pa_Y$. 
The $score(S)$ is defined as the weighted sum of  \textit{Independence} and \textit{Significance} terms. Let us first concentrate on the former. By definition, we write $\textit{Independence} = -\widehat{HSIC}(\textbf{X}_S, \hat{\varepsilon}_S)$. On a population level, it holds (Lemma~12 in \cite{reviewANMMooij}) that  ${HSIC}(\textbf{X}_S, {\varepsilon}_S) > 0 $ and ${HSIC}(\textbf{X}_{pa_Y}, {\varepsilon}_{pa_Y}) = 0$, since $\textbf{X}_S$ and ${\varepsilon}_S$ are not independent (because $S$ is not $\mathcal{F}$-plausible) and  $\textbf{X}_{pa_Y}$ and ${\varepsilon}_{pa_Y}$ are independent (by definition of the SCM). 
By Theorem~20 in \cite{reviewANMMooij}, we obtain $\widehat{HSIC}(\textbf{X}_{pa_Y}, \hat{\varepsilon}_{pa_Y})\to {HSIC}(\textbf{X}_{pa_Y}, {\varepsilon}_{pa_Y})=0$ and $\widehat{HSIC}(\textbf{X}_{S}, \hat{\varepsilon}_{S})\to {HSIC}(\textbf{X}_{S}, {\varepsilon}_{S})>0$, as $n\to\infty$. Therefore, the independence term is strictly smaller (for some large $n$) for $S$ than for $pa_Y$. 

Let us focus on \textit{Significance} term (We work with the \textit{Significance} term somewhat vaguely in this proof. However, we only need $Significance\to 0$ as $n\to\infty$ for $pa_Y$, which is satisfied for any reasonable method of assessing significance of covariates.). Since all $\textbf{X}_{pa_Y}$ are significant (otherwise $f_Y\notin\mathcal{I}_m$) we get that $Significance\to 0$ as $n\to\infty$ for $pa_Y$. Moreover, by definition, always $Significance\geq 0$.

Together, we find that $score(pa_Y)>score(S)$ for large $n$, since 
the \textit{Independence} term is strictly smaller (for large $n$) for $S$ than for $pa_Y$ and \textit{Significance} term converges to $0$ for $pa_Y$ and is non-negative.  
We showed that $pa_Y$ has the largest score among all $S\subseteq \{1, \dots, p\}$ (again, for $n$ large enough). 
\end{proof}

\end{document}